\documentclass[journal,letterpaper]{IEEEtran}

\usepackage{graphicx}
\usepackage{subfigure}
\usepackage{bm}
\usepackage{mathrsfs}
\usepackage{algorithm}
\usepackage{algorithmic}
\usepackage{color}
\usepackage{amsfonts}
\usepackage{multirow}
\usepackage{amssymb,amsthm}
\usepackage[cmex10]{amsmath}
\usepackage{epstopdf}
\usepackage{setspace}

\newcommand{\tabincell}[2]{\begin{tabular}{@{}#1@{}}#2\end{tabular}}
\newtheorem{theorem}{Theorem}
\newtheorem{corollary}[theorem]{Corollary}
\allowdisplaybreaks[4]

\begin{document}

\title{Small Cell Transmit Power Assignment Based on Correlated Bandit Learning}

\author{Zhiyang~Wang, 
	and~Cong~Shen,~\IEEEmembership{Senior Member,~IEEE}
\thanks{Z. Wang and C. Shen are with the Department of Electronic Engineering and Information Science, School of Information Science and Technology, University of Science and Technology of China, Hefei 230027, China. E-mail:  \texttt{wzy43@mail.ustc.edu.cn}, \texttt{congshen@ustc.edu.cn}.}
}

\maketitle

\begin{abstract}

Judiciously setting the base station transmit power that matches its deployment environment is a key problem in ultra dense networks and heterogeneous in-building cellular deployments. A unique characteristic of this problem is the tradeoff between sufficient indoor coverage and limited outdoor leakage, which has to be met without explicit knowledge of the environment. In this paper, we address the small base station (SBS) transmit power assignment problem based on stochastic bandit theory. Unlike existing solutions that rely on heavy involvement of RF engineers surveying the target area, we take advantage of the human user behavior with simple coverage feedback in the network, and thus significantly reduce the planned human measurement. In addition, the proposed  power assignment algorithms follow the Bayesian principle to utilize the available prior knowledge from system self configuration. To guarantee good performance when the prior knowledge is insufficient, we incorporate the performance \textit{correlation} among similar power values, and establish an algorithm that exploits the correlation structure to recover majority of the degraded performance. Furthermore, we explicitly consider \textit{power switching penalties} in order to discourage frequent changes of the transmit power, which cause varying coverage and uneven user experience. Comprehensive system-level simulations are performed for both single and multiple SBS deployment scenarios, and the resulting power settings are compared to the state-of-the-art solutions. Significant performance gains of the proposed algorithms are observed. Particularly, the correlation structure enables the algorithm to converge much faster to the optimal long-term power than other methods. 

\end{abstract}


\section{Introduction}
\label{sec:intro}

The massive deployment of distributed low-power low-cost small base stations (SBS) has been viewed as one of the most important solutions to address the challenge of exponential growth of the wireless data traffic, particularly for indoor users \cite{Cisco:16}. In practice, SBSs may be deployed in drastically different scenarios, from large warehouses and buildings to small residential apartments and single-office enterprises. In addition, the radio frequency (RF) conditions may vary significantly from one site to another. Due to the heterogeneous nature of these deployments, the transmit power assigned to the SBS, which effectively determines the coverage range, cannot be the same but must be decided based on the individual deployment environment, such as the building layout, the RF conditions, and the locations of the base stations.  Furthermore, indoor enterprise deployments often have stringent access and security constraints. As a result, judiciously setting the SBS transmit power to automatically match its deployment environment is among the most important challenges for in-building SBS network deployment \cite{Quek:13}.

To address this challenge, in-building enterprise networks typically rely on RF engineers to carry out extensive measurement and RF survey to determine the transmit power for appropriate coverage and limited leakage. Then, during live network operations, the RF engineers often need to make extra visits to optimize the transmit power for better performance. Clearly, this is a heavy human-in-the-loop model, as the success of the power setting relies on the experience of the seasoned engineers, the result of the RF survey of the engineers' choice, and the planning software. Not only is this approach expensive, inflexible and error-prone, but it also does not scale with the densification of indoor SBS networks \cite{Ramiro:11}.

Adaptive, automated and autonomous network optimization is the key principle of the self-organizing networks (SON) paradigm \cite{SON}, which aims at achieving the optimal network configuration while minimizing the planned human involvement in the deployment, configuration, optimization and maintenance. Self-optimizing the SBS transmit power falls into the framework of SON, and several solutions have already been proposed. Small Cell Forum has defined a common network monitor mode \cite{SmallCell}, allowing each SBS to periodically measure its surrounding RF environment and adjust its transmit power. This solution relies on an assumed coverage range based on categorization, and the RF measurements are only taken at the SBS location but not over the entire coverage area, which is coarse and may cause RF mismatch \cite{NLM}. To solve these issues, Supervised Mobile Assisted Range Tuning (SMART) was proposed in \cite{SMART}, which relies on the RF feedback of a technician walking along the sampling routes. The required RF feedback is extensive, including majority of the LTE lower layer quantities such as RSRP RSSI, CQI, etc. These quantities along the measurement routes provide important RF information of the deployment, and a global optimization can be formulated to derive the transmit power that satisfies both coverage and leakage constraints. Unfortunately, this problem is non-convex and the optimal transmit power is  difficult to compute \cite{SMART}. In \cite{Lopez:11}, the authors developed a self-organizing policy for distributed femtocell networks, aiming at minimizing the cell transmit power while satisfying the service requirement. In \cite{Sumeeth}, a heuristic solution was proposed to reliably determine the coverage for the current power level before either increasing or decreasing the power based on user feedback.  Solutions from both \cite{Lopez:11} and \cite{Sumeeth} have some adaptability but still lack good accuracy when used in different environments. The authors of \cite{Kim:14} modeled SBS power management  as a Markov Decision Process problem, focusing on the power control in a time-varying network. Similarly, a downlink transmit power control solution for interference mitigation via reinforcement learning was proposed in \cite{Bennis:13}. The main objective of \cite{Kim:14} and \cite{Bennis:13}, however, is to adjust the transmit power in reaction to the changing circumstance for better quality of service, which makes it more of a power control problem that has to be solved at a fast time scale.

We focus on setting the SBS transmit power of an enterprise network in an unknown deployment environment. We limit our attention to SBS networks with \textit{closed access} mode, which is commonly adopted in the enterprise deployment due to security and management considerations. An adequate power assignment is particularly crucial for the closed access mode, as the transmit power needs to be large enough to provide sufficient coverage for the inside users while small enough to not create significant interference to the outside non-enterprise co-channel users, who cannot be served by the enterprise network. This work proposes to capture this delicate balance between coverage and leakage by a system performance indication function (PIF). If the deployment is known, the optimal power assignment can be obtained by maximizing the PIF.

However, a practical solution needs to be effective in an arbitrarily unknown environment, and prefers minimum human involvement and feedback. Naturally, a good solution must compliment the aforementioned \textit{optimization} problem with an \textit{online learning} approach to remove the uncertainty of the environment, which is a key challenge for efficient transmit power assignment. The SBSs have to balance the immediate gains (selecting a power level that performs best so far) and long-term performance (evaluating other power levels). We thus resort to the theory of multi-armed bandit (MAB) \cite{Bubeck:12} to address the resulting exploration and exploitation tradeoff. However, as opposed to directly applying classical MAB algorithms such as UCB \cite{Auer:02}, our problem has two unique characteristics that were not exploited. First, SBS transmit power assignment falls into the \textit{self-optimization} category of SON. Generally, a \textit{self-configuration} phase has already taken place before invoking the transmit power assignment algorithm. As a result, there would be some \textit{prior knowledge} of the system that can be utilized. Second, performances of similar power levels are often very similar, which means that if we adopt the MAB model, nearby arms are highly \textit{correlated}. Intuitively, such correlation can be used to accelerate the convergence to the optimal selection, because any sampling of a power level not only reveals information about itself, but also nearby power levels that are highly correlated. Such information was not available in classical UCB solutions \cite{Bubeck:12,Auer:02}\footnote{The authors of \cite{Robert:04} studied  the continuum-armed bandit with an infinite continuum of strategies, which also captures the dependency among arms. We opt out this approach because in the multi-SBS cases, the continuity of the reward functions may not be guaranteed. The discrete arm setting makes the solutions more effective and flexible for practical adoption.}, and has not been utilized in SON \cite{SMART,Lopez:11,Kim:14,Bennis:13}.

In this paper, we leverage these engineering characteristics of the problem, and develop bandit-inspired transmit power assignment algorithms. In the bandit literature, similar models have been studied in \cite{Reverdy:14,Srivastava:15} and the corresponding bandit algorithms have been proposed. The authors of \cite{Reverdy:14} proposed bandit algorithms with a Bayesian prior on the mean reward that is based on a human decision-making model. \cite{Srivastava:15} further extended the algorithm to focus on the correlation among arms. In our work, we first adopt a Bayesian \cite{Kaufmann:12} learning algorithm that incorporates the prior knowledge of the system from the self-configuration phase. The developed \emph{Bayesian Power Assignment} (BPA) algorithm iteratively updates the posterior distribution based on new observations and the prior distribution, and uses the updated posterior distribution to compute the utility function and determine the transmit power level. In addition to utilizing the prior knowledge, we further leverage the correlation structure of the PIF of similar transmit power levels, and a \emph{Correlated Bayesian Power Assignment} (CBPA) algorithm that combines the Bayesian principle with the correlation property is employed. To the authors' best knowledge, this is the first work that incorporates \textit{bandit with correlated arms} into the design of wireless networks. Furthermore, practical deployment often wants to avoid frequent power changes, because it may cause frequent variation of the coverage area and result in uneven user experience. To address this issue, we present a block allocation extension to the proposed BPA and CBPA algorithms which explicitly considers switching cost to discourage frequent changes of power levels. Rigorous analysis of the performance loss with respect to the genie-aided global optimization solution is carried out. A tight upper bound of the performance loss for the most general algorithm (CBPA with switching cost) is derived, and performance characterization of other algorithms can be obtained as special cases. In order to reduce the algorithms' complexity which increases exponentially with the number of SBSs, we further introduce \textit{clustering} based on the prior knowledge, so that the complexity can be drastically reduced without sacrificing much of the accuracy and effectiveness of the algorithms.  The performances of all the proposed algorithms are verified by extensive system-level simulations and compared with both the globally optimal power assignment with complete information and the existing state-of-the-art solutions. Not only do the proposed algorithms outperform existing solutions and converge to the globally optimal power assignment quickly, but they also reduce the planned human involvement significantly and only require minimum amount of user feedback (one bit per location), as opposed to the full-blown RF measurement and feedback that is universally required in the existing solutions.

The rest of the paper is organized as follows. The system model and problem formulation can be found in Section~\ref{sec:sys}. Section~\ref{sec:alg} and \ref{sec:alg2} present the proposed power assignment algorithms without and with switching cost, respectively. Performance analysis for all the algorithms is given in Section~\ref{sec:analysis}. Complexity issues of the multi-SBS deployment are addressed in Section~\ref{sec:multi}. Simulation results are portrayed in Section~\ref{sec:sim}. Finally, Section~\ref{sec:conc} concludes the paper.

\section{System Model and Problem Formulation}
\label{sec:sys}
\subsection{Network Model}

Both single-SBS and multi-SBS deployments are considered. Note that the former is suitable for modeling single-office enterprises, residential apartments and other small {{deployments}}, while the latter mainly applies to large enterprises, for which multiple SBSs are installed to jointly cover the indoor users. The set of SBSs is indexed as $\mathcal{K}_{SBS}=\{1,2,..,K\}$. Each SBS has a set of candidate pilot\footnote{As the purpose of the long-term power assignment is to determine the appropriate coverage that fits the deployment, we focus on setting the pilot power instead of the power of data and control channels \cite{Quek:13}.} power levels, denoted as $\mathcal{P}=\{p_1,p_2,..,p_n\}$. As our focus is on the SBSs with closed access and co-channel with the macro base stations (MBS), we simply assume that the users at the measurement points inside the enterprise building are served by the SBS network, while users at points outside can only be served by one of the MBSs from $\mathcal{K}_{MBS}=\{1,2,..,K_M\}$, as Fig.~\ref{deploy} illustrates. 

The measurement data come from the customer UE feedback from some inside and outside routes during normal network operations. This is different from the RF survey approach that is carried out during network planning. The detailed mechanism and procedure of obtaining such customer UE feedback are mostly the same as in \cite{Sumeeth}. However, as opposed to a complete RF  feedback required in \cite{Sumeeth}, we only require \textit{one-bit} coverage indication for each inside report. The extended set of RF measurements, such as RSRP, RSSI, and CQI, are not needed in our power assignment algorithm. For non-enterprise UEs, as we only need to know whether the UE is covered at a reporting location, we will rely on the \textit{registration attempt} at the outside location to determine such events. Note that this is a common approach to determine leakage and has been adopted in \cite{Claussen:08,SMART,25304}.

\begin{figure}[h]
        \begin{center}
        \subfigure[ Single-SBS ]{
        \includegraphics[width=0.4\textwidth]{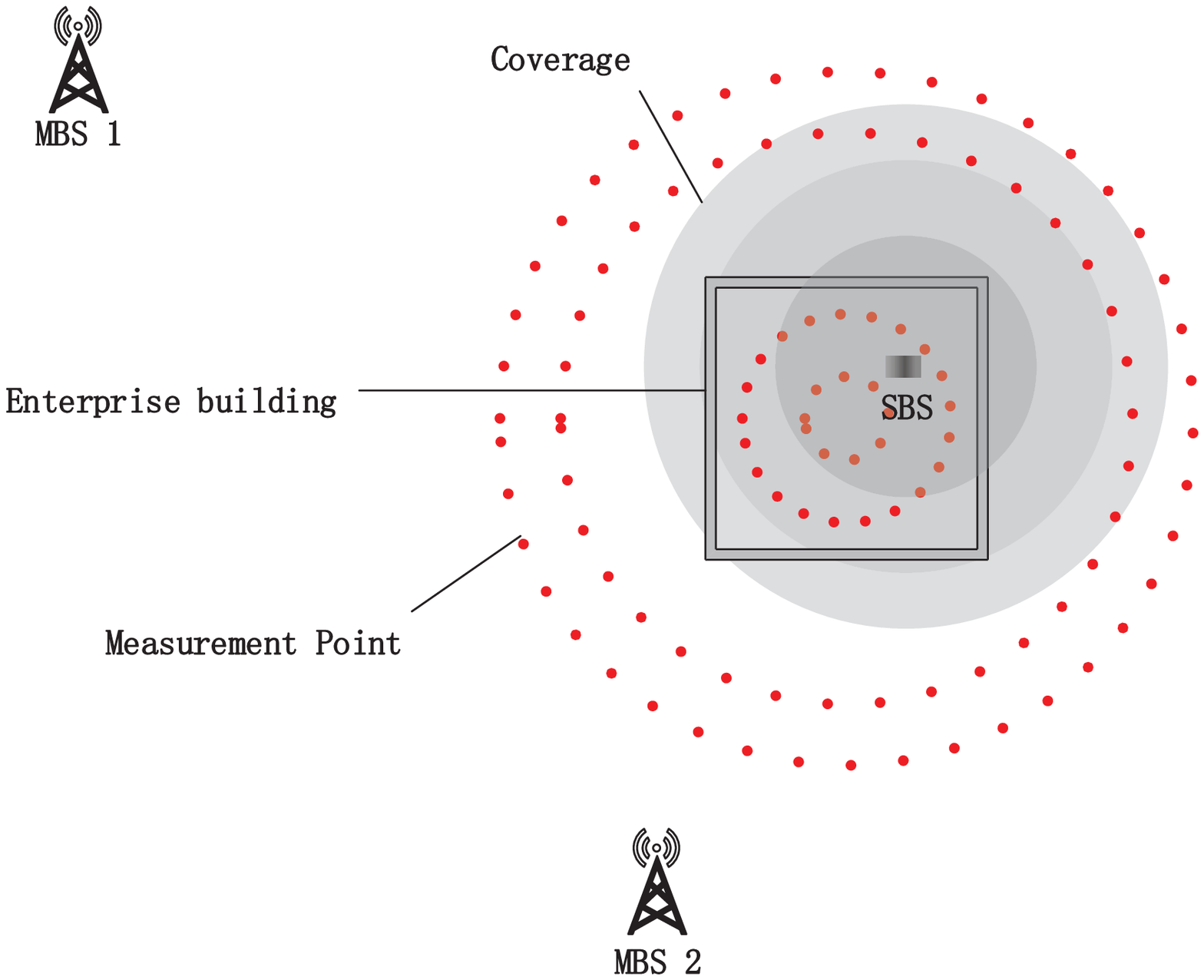}
        \label{fig:sinBS} } \end{center}
        \hfil
        \subfigure[ Multi-SBS ]{
        \includegraphics[width=0.4\textwidth]{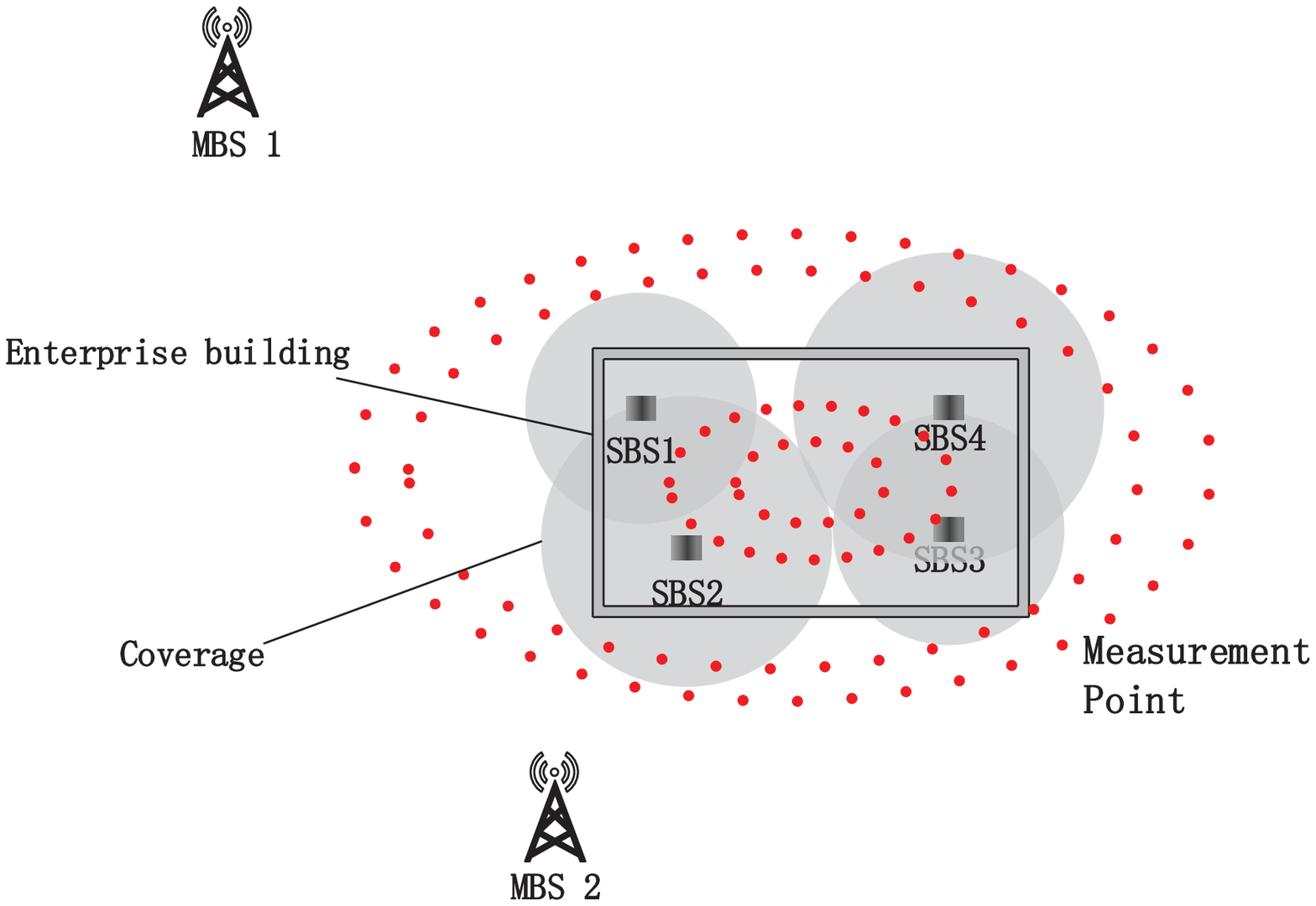}
        \label{fig:mulBS}}
    \caption{An exemplary enterprise SBS network deployment.}
    \label{deploy}
\end{figure}

In this work, our model and procedure on power assignment follow the common industry SON operations \cite{Ramiro:11}. Specifically, the power assignment policy is executed during the \textit{self-optimization} phase of SON, at the central network controller which is configured to oversee the operation of the entire SBS network. This is a common choice for {enterprise} cellular networks, as they often have security and privacy constraints which are easier to be satisfied in a centralized architecture. Furthermore, the power assignment algorithm operates in a periodic fashion, which is typical for self-optimization of SON \cite{Claussen:08}. For each time slot, the SBS first sets the pilot power based on the  assignment algorithm. Then the network operates and collects UE feedback from both inside and outside of the intended coverage area. At the end of the current period, a performance measure is computed to evaluate the current pilot power and then used in the assignment algorithm to compute the power level for the next slot. This sequence of operations is illustrated in Fig.~\ref{timeslot}. Lastly, industry SON operations typically have the \textit{self-optimization} operations follow a \textit{self-configuration} phase, during which a coarse measurement and power calibration are performed \cite{SMART}. As we will see later, the initial self-configuration, albeit coarse and sometimes inaccurate, offers useful prior knowledge that can be leveraged in the power assignment algorithm.

\begin{figure}[h]
\centering\includegraphics[width=0.48\textwidth]{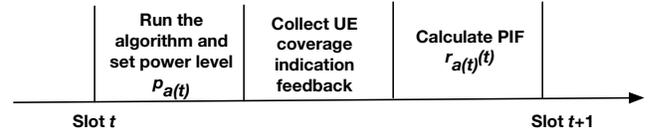}
\caption{The power assignment procedure in a time slot $t$.}
\label{timeslot}
\end{figure}

\subsection{Problem Formulation}
\label{sec:formu}

To formulate the power assignment problem, we first need to define the criteria for coverage and leakage. To that end, let us denote the set of measurement points on the inside and outside routes as $N_{in}=\{1,2,..,n_{in}\}$ and $N_{out}=\{1,2,..,n_{out}\}$, respectively. The coverage and leakage criteria for a measurement point can be formally defined as:
\begin{eqnarray}
\text{coverage: } & \max\limits_{k_S\in \mathcal{K}_{SBS}}{{\sf{SINR}}_{k_S,n}} > {{\sf{SINR}}_{\text{th}}}, \text{for } n \in N_{in}, \label{eqn:cov_def} \\
\text{leakage: } & \max\limits_{k_M\in \mathcal{K}_{MBS}}{{\sf{SINR}}_{k_M,n}} < {{\sf{SINR}}_{\text{th}}},  \text{for }  n\in N_{out}, \label{eqn:leak_def}
\end{eqnarray}
where ${{\sf{SINR}}_{k_S,n}}$ and ${{\sf{SINR}}_{k_M,n}}$ represent the SINR of the measurement point $n$ inside corresponding to SBS $k_S$ and the SINR of point $n$ outside served by MBS $k_M$, respectively. They can be calculated as:
 \begin{eqnarray*}
 {{\sf{SINR}}_{k_S,n}}&=&\frac{P^r_{k_S,n}}{\sum_{i=1,i\neq k_S}^{K}{P^{r}_{i,n}}+ \sum_{j=1}^{K_M}P^{Mr}_{j,n}+N_{s}}, \\
 {{\sf{SINR}}_{k_M,n}}&=&\frac{P^{Mr}_{k_M,n}}{\sum_{i=1}^{K}{P^{r}_{i,n}} + \sum_{j=1,j\neq k_M}^{K_M}P^{Mr}_{j,n}+N_{s}},
 \end{eqnarray*}
where $P^{r}_{i,n}$ and $P^{Mr}_{j,n}$ represent the received power at point $n$ from SBS $i$ and MBS $j$, respectively, $N_{s}$ denotes the uncontrolled noise and interference, and  ${{\sf{SINR}}_{\text{th}}}$ is the SINR threshold.

With the definition at each measurement point, the overall system coverage and leakage are defined as the percentage of measurement points which satisfy the coverage condition \eqref{eqn:cov_def} and leakage condition \eqref{eqn:leak_def}, respectively. If we denote the number of measurement points that satisfy the corresponding conditions as $n_{cov}$ and $n_{lea}$, then the coverage percentage and leakage percentage can be computed as $\eta_{in}=n_{cov}/n_{in}\times 100\%$ and $\eta_{out}=n_{lea}/n_{out}\times 100\%$. Note that a larger pilot transmit power may simultaneously increase the indoor coverage percentage and the outdoor leakage percentage. Hence, the system performance indication function (PIF) associated with each candidate power level must balance coverage and leakage. In this work, we adopt a simple linear PIF as
\begin{equation}
\label{eqn:PIF}
r=\alpha\eta_{in}-(1-\alpha)\eta_{out},
\end{equation}
where {{$\alpha\in[0,1]$}} is a control parameter and can be tuned to weigh differently between coverage and leakage. Note that PIF \eqref{eqn:PIF} is chosen as an example to illustrate the proposed power assignment algorithms. Other meaningful PIFs that capture the tradeoff between coverage and leakage can be used in place of \eqref{eqn:PIF}.
The objective of a power assignment algorithm is to find the optimal solution $p^* \in \mathcal{P}$ that maximizes the PIF \eqref{eqn:PIF}. 

Strictly speaking, the function $r$ in \eqref{eqn:PIF} is a random variable for a given pilot power level. This is due to the random channel effect such as shadowing, fast fading and other disturbance in the deployment environment. We focus on a probabilistic model with \textit{Gaussian} random fluctuation around the mean.  As we will see in Sec.~\ref{sec:sim}, Gaussian distribution indeed is a very good approximation for the actual PIF. Furthermore, we evaluate the proposed algorithms in settings with non-Gaussian PIF distributions, and the empirical results suggest that algorithms developed based on the Gaussian assumption are very effective.

\section{Power Assignment Algorithms based on Bayesian Bandit Learning}
\label{sec:alg}

\subsection{Stochastic Bandit Model}
\label{sec:bandit}

The necessity of balancing the short-term performance and long-term learning has motivated us to take a stochastic multi-armed bandit approach to the power assignment problem. Specifically, we model the set of candidate pilot power values $\mathcal{P}=\{p_1,p_2,..,p_n\}$ as $n$ arms, denoted by $\mathcal{N}_{pow}=\{1,2,..,n\}$. At the beginning of each time slot $t=1,2,..,T$, a power value $p_{a(t)}\in\mathcal{P},a(t)\in\mathcal{N}_{pow}$ is selected. At the end of the time slot $t$, the SBS observes a performance feedback $r_{a(t)}(t)$ based on UE measurement reports, which corresponds to $reward$ in the bandit theory. As discussed in Sec.~\ref{sec:sys}, we model the random PIF associated with \textit{each} power value as a Gaussian random variable. The objective is to develop an efficient power assignment solution to maximize the cumulative PIF for any given time horizon $T$. For the multi-SBS case, each arm corresponds to a set of power levels of all $K$ SBSs, and other definitions remain the same.

In multi-armed bandit theory, a quantity termed as \emph{expected cumulative regret} \cite{Bubeck:12} is often used to characterize the algorithm performance, which represents the cumulative difference between the reward of the arms chosen and the maximum expected reward, which is attainable by a ``genie'' who knows the expected reward of all arms. We comment that minimizing the expected cumulative regret is equivalent to maximizing the expected accumulated reward, which is the objective of the power assignment problem. This is because the maximum expected reward is independent of the adopted learning algorithm and the regret is equivalent to the performance loss of any power assignment problem due to learning.

Formally,  we denote
\begin{equation}
\label{eqn:cumR}
G_T=\sum\limits^{T}\limits_{t=1} r_{a(t)}(t)
\end{equation}
as the cumulative PIF up to a given time horizon $T>0$, and we define the cumulative PIF loss \textit{due to learning} as
\begin{eqnarray}
\label{eqn:regret}
R_{T} =G_T^{*}-G_T  = \max\limits_{i=1,..,n}\left(\sum\limits^{T}\limits_{t=1}r_{i}(t)\right)-\sum\limits^{T}\limits_{t=1}r_{a(t)}(t),
\end{eqnarray}
which corresponds to the definition of cumulative regret. Here the optimal power level can be obtained by a genie-aided solution, e.g. a global optimization of the expected PIF with complete RF information from the technician survey. We are interested in finding efficient algorithms that maximize the cumulative PIF \eqref{eqn:cumR}. Equivalently, the goal is to minimize the PIF loss of the system \eqref{eqn:regret} for any given time horizon $T$. The expected PIF loss can be written as:
\begin{eqnarray}
\mathbb{E}[R_{T}]&=&T\mu^*-\mathbb{E}\sum\limits^T\limits_{t=1}\mu_{a(t)} \nonumber \\
&=&\left(\sum\limits_{i=1}\limits^n\mathbb{E}[N_i(T)]\right)\mu^*-\mathbb{E}\sum\limits_{i=1}\limits^n N_i(T)\mu_i \nonumber \\
&=& \sum\limits_{i=1}\limits^n\Delta_i\mathbb{E}[N_i(T)],
\label{eqn:regdelta}
\end{eqnarray}
where $\mu^{*}=\max\limits_{i=1,..,n}\mu_{i}$ is the true mean PIF of the optimal power level and $\Delta_{i}=\mu^{*}-\mu_{i}$ measures the mean PIF gap between the chosen power level and the optimum. $N_i(T)$ represents the number of times power level $p_i$ is selected. According to the ground-breaking work  of Lai and Robbins \cite{LaiRobbins}, if the expected loss $\mathbb{E}[R_{T}]$ of our proposed algorithms can be upper bounded\footnote{$\log(\cdot)$ represents natural logarithm if the base is not specified.} by $\mathcal{O}(\log T)$, an asymptotically optimal performance is achieved in the sense that  the convergence rate is of the same order as the optimum.

\subsection{Bayesian Power Assignment Algorithm}
\label{sec:BPA}

The first algorithm utilizes the \textit{prior} knowledge of the PIF estimation \textit{before} the algorithm is invoked. In practice, the most common form for the prior knowledge comes from the self-configuration phase of SON, which is performed during network initialization. This phase can provide us with some prior estimation of the PIFs as it typically tries different power levels before settling on one. However, all practical SON solutions have certain requirements on the elapsed time of the self-configuration operations. This is because self-configuration affects the boot-up time, and thus must be carefully controlled. As a result, massive measurement during self-configuration is typically out of the question and we often encounter  coarse initial setup. Another possibility is that as the proposed power assignment algorithm is recursive  over time, it also progressively collects PIF estimations for each selected power level. This can be used iteratively to update the prior knowledge. {The quality of the prior depends on the detailed process in self-configuration phase, e.g. the time duration, mechanisms for large power settings, which is uncontrollable and out of scope of this paper. However, }it is worth noting that the proposed algorithms also work {with inaccurate prior or even }without any prior knowledge, at the expense of slower convergence.

We first consider the power assignment algorithm without considering the correlation between power levels. We adopt the well-known Bayesian principle \cite{Kaufmann:12} that integrates the prior distribution and quantiles of the posterior distribution. {{The proposed \textit{Bayesian Power Assignment (BPA)} algorithm, which adopts the deterministic \textit{upper credible limit (UCL)} principle in \cite{Reverdy:14}, is given in Algorithm~\ref{alg:BPA}.}} In this algorithm, $\{\mu_{i}^{0}, \sigma_{0}^{2}\}$ denotes the prior knowledge of the Gaussian distribution for PIF. The utility function defined in step~\ref{alg:Qbpa} is composed of an estimated performance term and a measure of uncertainty, which reflects the tradeoff between exploration and exploitation. More specifically, $\Phi^{-1}:(0,1)\rightarrow\mathbb{R}$ is the inverse cumulative distribution function (CDF) for a standard Gaussian random variable. We use the quantile function to indicate: $\mathbb{P}(\mu_{i}\leq Q_{i}^{\text{BPA}}(t))=1-1/(\sqrt{2\pi e}t^2)$. Asymptotically, the true mean PIF $\mu_{i}$ is more likely to be less than the estimation $Q_{i}^{\text{BPA}}$, which leads to the convergence to the optimal power level.

\begin{algorithm}
\caption{The Bayesian Power Assignment (BPA) Algorithm}
\label{alg:BPA}
\begin{algorithmic}[1]
\begin{spacing}{1.3}
\REQUIRE Prior estimation of PIF mean $\{\mu_{i}^{0}\}_{i=1}^{n}$, variance $\sigma_{0}^{2} $ 
\ENSURE  $N_{i}(t)=0, \bar{r}_{i}(t)=0, Q^{BPA}_{i}(t)=0, \hat{\mu}_i(1)=\mu_i^0, \hat{\sigma}_{i}(1)=\sigma_0$ for all $i\in\mathcal{N}_{pow}$, $t\in{1,..,T}$.

\FOR{$t \in {1,2,..,T}$}
\STATE For each arm $i \in\mathcal{N}_{pow}$ update the utility function:
\\$Q^{\text{BPA}}_{i}(t)= \hat{\mu}_{i}(t)+\hat{\sigma}_{i}(t)\Phi^{-1}(1-1/(\sqrt{2\pi e}t^2))$, \label{alg:Qbpa}
\STATE Select a power value $p_{a(t)}$ according to:
\\$a(t)=\arg\max\{Q^{\text{BPA}}_{i}(t)|i \in \mathcal{N}_{pow}\}$,
\STATE Observe the PIF $r_{a(t)}(t)$,
\STATE Update the average PIF and the selected times of $p_{a(t)}$:
\\$\bar{r}_{a(t)}(t+1)=\frac{N_{a(t)}(t)\bar{r}_{a(t)}(t)+r_{a(t)}(t)}{N_{a(t)}(t)+1}$,
\\ $N_{a(t)}(t+1) = N_{a(t)}(t)+1$,
\STATE Update the estimated mean and variance of PIF of power level $p_{a(t)}$:
\\$\hat{\mu}_{a(t)}(t+1)=\frac{\mu_{a(t)}^{0}+N_{a(t)}(t+1)\bar{r}_{a(t)}(t+1)}{N_{a(t)}(t+1)+1}$,
\\$\hat{\sigma}_{a(t)}(t+1)=\frac{\sigma_{0}}{\sqrt{N_{a(t)}(t+1)+1}}$. \label{alg:cpaupdate}
\ENDFOR
\end{spacing}
\end{algorithmic}
\end{algorithm}

If the prior knowledge is not available, the BPA algorithm can be slightly modified to address this issue. Specifically, the estimated PIF term and uncertainty measurement have to be updated simultaneously in each time slot. This philosophy leads to the following utility function:
\begin{equation}
\label{eqn:uipaQ}
\begin{split}
& Q^{{\text{UiPA}}}_{i}(t) =\bar{r}_{i}(t)+ \\
& \sqrt{\frac{\sum\limits_{\tau=1}\limits^{t}{r_{i}^2(\tau)}-\bar{r}_i^2(t)N_{i}(t)}{(N_{i}(t)-1)N_{i}(t)} }\Phi^{-1}(1-1/(\sqrt{2\pi e}t^2)).
\end{split}
\end{equation}
The \textit{Uninformative Power Assignment (UiPA)} algorithm thus can be obtained by replacing the utility function in step \ref{alg:Qbpa} of Algorithm~\ref{alg:BPA} with (\ref{eqn:uipaQ}), while removing the prior input at the beginning and estimation state update in step~\ref{alg:cpaupdate}.

\subsection{Correlated Bayesian Power Assignment Algorithm}

In the BPA algorithm, $\{\mu_{i}^{0}, \sigma_{0}^{2}\}$ is used as our prior knowledge of performance for each power level. If the PIFs of different arms are independent, then utilizing individual Gaussian distributions is sufficient in our framework. However, for the considered power assignment problem, the PIFs of similar transmit power levels are generally correlated due to the slow and continuous changing nature of RF propagation. In other words, a stronger PIF correlation exists between adjacent power levels than distant pairs, and leveraging the full covariance matrix of the joint distribution may provide significant performance boost compared to the BPA algorithm. Intuitively, if a transmit power level results in a bad PIF with respect to the balance of coverage and leakage, then an intelligent algorithm may not need to waste much exploration on its immediate neighboring power levels, as they are highly likely to be bad as well.

\begin{algorithm}
\caption{The Correlated Bayesian Power Assignment (CBPA) Algorithm}
\label{alg:CBPA}
\begin{algorithmic}[1]
\begin{spacing}{1.3}
\REQUIRE Prior estimation of joint Gaussian distribution of the PIFs: $\mathcal{N}(\bm{\mu}_{0},\Sigma_{0}) $;
\ENSURE $N_{i}(t)=0, \bar{r}_{i}(t)=0, Q^{CBPA}_{i}(t)=0, \hat{\mu}_{i}(1)=\mu_{i}^0, \hat\Sigma_1=\Sigma_0$ for all $i\in\mathcal{N}_{pow}$ and $t\in{1,..,T}$.
\FOR{$t \in {1,2,..,T}$}
\STATE For each $i\in\mathcal{N}_{pow}$ update the utility function
\\$Q^{\text{CBPA}}_{i}(t)=\hat{\mu}_{i}(t)+\hat{\sigma}_{i}(t) \sqrt{\sum \limits_{j=1} \limits^{n} \rho_{ij}^2(t)} \Phi^{-1}(1-1/(\sqrt{2\pi e} t^2))$,\label{alg:Qcbpa}
where $\rho_{ij}(t)$ is the correlation coefficient between power value $i$ and $j$ {{at time $t$, which is obtained from $\hat{\Sigma}_t$}}; $\hat{\mu}_{i}(t)$,$\hat{\sigma}^2_{i}(t)$ is the $i$-th entry of $\hat{\bm{\mu}}_{t}$ and diagonal entry of $\hat{\Sigma}_{t}$.
\STATE Select a power value $p_{a(t)}$ according to:
\\$a(t)=\arg\max\{Q^{\text{CBPA}}_{i}(t)|i \in \mathcal{N}_{pow}\}$,
\STATE Collect the performance function $r_{a(t)}(t)$,
\STATE Update the average performance and the selected time of $p_{a(t)}$:
 \\$\bar{r}_{a(t)}(t+1) = \frac{N_{a(t)}(t)\bar{r}_{a(t)}(t)+r_{a(t)}(t)}{N_{a(t)}(t)+1}$,
 \\$N_{a(t)}(t+1) = N_{a(t)}(t)+1$,
\STATE Update the estimation state: \\$\hat{\bm{\mu}}_{t+1}=(\Sigma_{0}^{-1}+P(t+1)^{-1})^{-1}(P(t+1)^{-1}\bar{\mathbf{r}}_{t+1}+\Sigma_{0}^{-1}\bm{\mu}_{0})$,
\\$ \hat{\Sigma}_{t+1}^{-1}=\Sigma_{0}^{-1}+P(t+1)^{-1}$. \label{alg:cbpaupdate}
\ENDFOR
\end{spacing}
\end{algorithmic}
\end{algorithm}

We formally present the \textit{Correlated Bayesian Power Assignment (CBPA)} algorithm in Algorithm~\ref{alg:CBPA}. Let $\mathcal{N}(\bm{\mu}_{0},\Sigma_{0})$ be a correlated prior assumption while $\Sigma_{0}$ is a positive definite matrix, we define $\{\bm{\phi}_{t}\in \mathbb{R}^{n}\}_{t\in \{1,..,T\}}$ as the indicator vector to reveal the currently selected power value $p_{a(t)}$, i.e.,
\begin{equation*}
({\bm{\phi}}_{t})_{k}=
\begin{cases}
1& \text{$k=a(t)$,}\\
0& \text{otherwise,}
\end{cases}
\end{equation*}
where $({\bm{\phi}}_{t})_{k}$ represents the $k$-th entry of $\bm{\phi}_{t}$.
The estimation of the mean PIFs and correlation structure of the PIF (${\bm{\mu}_{t},\Sigma_{t}}$) is updated following the Bayesian principle \cite{Srivastava:15}:
\begin{equation*}
\begin{aligned}
&\mathbf{q}_t=\frac{r_{t}\bm{\phi}_{t}}{\sigma_{0}^{2}}+\hat{\Lambda}_{t-1}\hat{\bm{\mu}}_{t-1},
&\hat{\Lambda}_{t}=\frac{\bm{\phi}_{t}\bm{\phi}_{t}^{T}}{\sigma_{0}^{2}}+\hat{\Lambda}_{t-1},
\\&\hat{\Sigma}_{t}=\hat{\Lambda}^{-1}_{t}, &\hat{\bm{\mu}}_{t}=\hat{\Sigma}_{t}\mathbf{q}_t=\hat{\Lambda}_{t}^{-1}\mathbf{q}_t,
\end{aligned}
\end{equation*}
where $r_{t}$ is the PIF observed at time slot $t$. To derive a general expression of the estimation, we introduce a diagonal matrix $P(t)$ with entries $\sigma_{0}^2/N_{i}(t), i\in\mathcal{N}_{pow}$, and $\bar{\mathbf{r}}_{t}$ is the vector of $\bar{r}_{i}(t), i\in\mathcal{N}_{pow}$. We first rewrite the expression of $\hat{\Lambda}_{t}$ as:
\begin{equation}
\label{lambda}
\begin{aligned}
\hat{\Lambda}_{t}&=\frac{\bm{\phi}_{t}\bm{\phi}_{t}^{T}}{\sigma_{0}^{2}}+\frac{\bm{\phi}_{t-1}\bm{\phi}_{t-1}^{T}}{\sigma_{0}^{2}}+\hat{\Lambda}_{t-2} \\
& =\frac{\bm{\phi}_{t}\bm{\phi}_{t}^{T}}{\sigma_{0}^{2}}+\frac{\bm{\phi}_{t-1}\bm{\phi}_{t-1}^{T}}{\sigma_{0}^{2}}+\ldots+\frac{\bm{\phi}_{1}\bm{\phi}_{1}^{T}}{\sigma_{0}^{2}}+\Lambda_{0}\\
           &=\frac{1}{\sigma_{0}^2}
           \begin{pmatrix}
           {N_{1}(t)} &\quad &\quad & \quad \\
           \quad & {N_{2}(t)} & \quad &\quad\\
           \quad & \quad  & \ddots & \quad\\
           \quad & \quad & \quad & {N_{n}(t)}
           \end{pmatrix}
           +\Lambda_{0} \\ 
           & =P(t)^{-1}+\Lambda_{0}.
\end{aligned}
\end{equation}
Then, $\hat{\bm{\mu}}_{t}$ can be derived based on (\ref{lambda}):
\begin{equation}
\label{mu}
\begin{aligned}
\hat{\bm{\mu}}_{t}&=\hat{\Lambda}_{t}^{-1}\mathbf{q}_{t}=\hat{\Lambda}_{t}^{-1}\left(\frac{r_{t}\bm{\phi}_{t}}{\sigma_{0}^2}+\hat{\Lambda}_{t-1}\hat{\Sigma}_{t-1}\bm{q}_{t-1}\right)\\
            &=\hat{\Lambda}_{t}^{-1}\left(\frac{r_{t}\bm{\phi}_{t}}{\sigma_{0}^2}+\frac{r_{t}\bm{\phi}_{t}}{\sigma_{0}^2}+\ldots+\frac{r_{t}\bm{\phi}_{t}}{\sigma_{0}^2}+\Lambda_{0}\bm{\mu}_{0}\right)\\
            &=\hat{\Lambda}_{t}^{-1}
            \left(
            \begin{pmatrix}
            \frac{N_{1}(t)}{\sigma_{0}^2}\bar{r}_{1}(t) &\quad &\quad \\
            \quad  & \ddots & \quad\\
            \quad &  \quad &\frac{N_{n}(t)}{\sigma_{0}^2}\bar{r}_{n}(t)
            \end{pmatrix}
            +\Lambda_{0}\bm{\mu}_{0}
            \right)\\
            &=(\Lambda_{0}+P(t)^{-1})^{-1}(P(t)^{-1}\bar{\mathbf{r}}_{t}+\Lambda_{0}\bm{\mu}_{0}).
\end{aligned}
\end{equation}
Finally, combining equation \eqref{lambda} and \eqref{mu}, the estimation at time slot $t$ can be written as:
\begin{equation*}
\begin{aligned}
\hat{\Lambda}_{t}&=P(t)^{-1}+\Lambda_{0}, \\
\hat{\bm{\mu}}_{t}&=(\Lambda_{0}+P(t)^{-1})^{-1}(P(t)^{-1}\bar{\mathbf{r}}_{t}+\Lambda_{0}\bm{\mu}_{0}),
\end{aligned}
\end{equation*}
which is used in Algorithm~\ref{alg:CBPA}.

\section{Power Assignment with Switching Cost}
\label{sec:alg2}

\subsection{Problem Formulation with Switching Cost}

In practice, it is very critical for any practical cellular deployment to avoid frequent power changes. In a cellular network, coverage variation due to the change of transmit power often results in poor user experience (call drop, low data rate, frequent handover, etc.), which in turn degrades the network performance significantly. To address this problem, we explicitly add a switching cost when the power level changes. In this way, a good power assignment policy will determine the optimal power value while minimizing frequent switches. We adopt a general switching loss function $s_{ij}=f(|p_{i}-p_{j}|)$, which is a bounded non-decreasing function of the difference between the two power values {{with $f(0)=0$}}. $s_{ij}$ is incurred whenever SBS changes its pilot power value between $p_j$ and $p_i$. The cumulative switching cost up to $T$ can be written as:
\begin{equation*}
{\sf{SC}}(T)=\sum\limits_{t=2}\limits^T s_{a(t)a(t-1)}=\sum\limits_{t=2}\limits^T f(|p_{a(t)}-p_{a(t-1)}|).
\end{equation*}
Thus the cumulative PIF in this problem can be expressed as:
\begin{equation*}
G^S_T=G_T-{\sf{SC}}(T).
\end{equation*}
In a multi-SBS deployment, the switching cost is defined as the sum of individual switching costs of all SBSs.

\subsection{The Power Assignment Algorithm with Switching Cost}
\label{sec:PASCalg}

\begin{figure*}
\centering\includegraphics[width=0.95\textwidth]{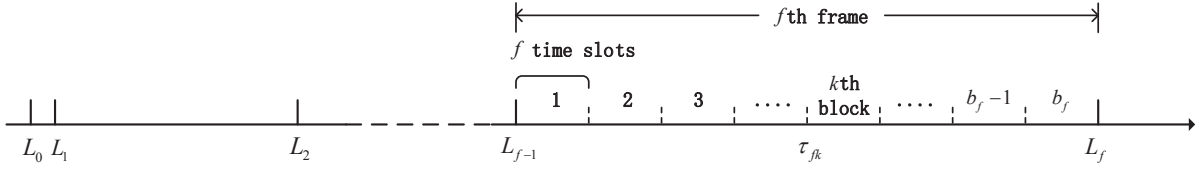}
\caption{The block allocation scheme used in BPA-SC and CBPA-SC.}
\label{blockscheme}
\end{figure*}

We extend the preceding algorithms to a \textit{block allocation} scheme to address switching costs. Block allocation schemes, such as the one in \cite{Agrawal:88}, determine specific intervals of time over which the selection is consistent. A power value is selected at the beginning of each interval. The construction of the intervals should ensure the expected number of switches scales at most logarithmically in time to guarantee good performance. This idea is graphically presented in Fig.~\ref{blockscheme}. We first divide time into frames whose last time slot is denoted as $L_f,f\in\{1,2...l\},l=\lceil \sqrt{\log_2T} \rceil$. Each frame is then subdivided into $b_f=\lceil  ( 2^{f^2}-2^{(f-1)^2} ) / f  \rceil$ blocks each of which contains $f$ time slots. Each block is identified by $(f,k),f\in\{1,2,..,l\},k\in\{1,2,..,b_f\}$, with $f$ and $k$ representing the frame number and block number within the frame respectively. The beginning time slot of block $k$ in the $f$-th frame is denoted as $\tau_{fk}$. Note that the key element in selecting the blocking length is to only incur ${o}(\log T)$ switching cost. In this way, the $\mathcal{O}(\log T)$ regret of the standard algorithm still dominates the total regret.

\begin{algorithm}
\caption{The Power Assignment with Switching Cost Algorithm}
\label{alg:scUCL}
\begin{algorithmic}[1]
\begin{spacing}{1.3}
\REQUIRE Prior estimation of PIF mean: $\mathcal{N}(\bm{\mu}_{0},\Sigma_{0}) $;
\ENSURE  $N_{i}(t)=0, \bar{r}_{i}(t)=0, Q_{i}(t)=0, \hat{\mu}_i(1)=\mu_i^0, \hat{\Sigma}_1=\Sigma_0$ for all $i\in\mathcal{N}_{pow}$, $t\in{1,..,T}$.
\FOR{$f \in \{1,2,..,l\}$}
\FOR{$k \in \{1,2,..,b_f\}$}
\STATE The beginning time slot of $k$-th block in the $f$-th frame $\tau_{fk}=L_{f-1}+1+f(k-1)$,
\STATE For each $i\in\mathcal{N}_{pow}$ update the utility function $Q_{i}$, \label{alg:Qsc}
\STATE Select a power value $p_{a^*}$ according to:
\\$a^*=\arg\max\{Q_{i}|i \in \mathcal{N}_{pow}\}$,
\STATE Keep SBS on power value $p_{a^*}$ for the next $(n_f-1)$ slots,
\STATE Collect the performance function $r_{a^*}(t)$, possibly excluding a switching loss $s_{a(t)a(t-1)}$, $t \in \{\tau_{fk},\tau_{fk}+1,..,T_e\}$, $T_e=\tau_{fk}+f-1,n_f=f$ if $\tau_{fk}+f-1 \leqslant T$, otherwise $T_e=T,n_f=T-\tau_{fk}+1$;
\STATE Update the average performance and the selected time of $p_{a^*}$:
 \\$\bar{r}_{a^*} = \frac{N_{a^*}\bar{r}_{a^*}+\sum_{t=\tau_{fk}}^{T_e}(r_{a^*}(t)-s_{a(t)a(t-1)})}{N_{a^*}+n_f}$,
 \\$N_{a^*} = N_{a^*}+n_f$,
\STATE Update the estimation state. \label{alg:estista}
\ENDFOR
\ENDFOR
\end{spacing}
\end{algorithmic}
\end{algorithm}

The \emph{Power Assignment with Switching Cost} algorithm is formally presented in Algorithm~\ref{alg:scUCL}. The \emph{Uninformative} (UiPA-SC), \emph{Bayesian} (BPA-SC) and \emph{Correlated Bayesian Power Assignment with Switching Cost} (CBPA-SC) algorithms can be similarly obtained, by replacing $Q_i$ with $Q^{\text{UiPA}}_i$, $Q^{\text{BPA}}_i$ and $Q^{\text{CBPA}}_i$ respectively. Note that in BPA-SC, the prior estimation state $\Sigma_0$ becomes a diagonal matrix with entries $\sigma_0^2$, while there is no prior input in UiPA-SC. At the beginning of each block, a power value is selected and the SBS locks on this power value in each of the next $f$ time slots in the block. The estimation update in step~\ref{alg:estista} also follows step~\ref{alg:cpaupdate} in Algorithm~\ref{alg:BPA} and step~\ref{alg:cbpaupdate} in Algorithm~\ref{alg:CBPA}.

There are two key ideas of Algorithm~\ref{alg:scUCL}. The first is that since the switching cost results in a penalty in performance, the algorithm needs to ``explore in bulk''. This is done by grouping time slots and not switching within these slots. The second is that as time goes by, the algorithm has more  information about the optimal power value, and hence the block size should increase to take advantage of the better knowledge.

\section{Performance Analysis of the Proposed Algorithms}
\label{sec:analysis}

So far, we have presented two sets of power assignment algorithms (without and with switching cost), each of which further consists of components that have different assumptions on the prior knowledge and the correlation structure. In this section, we will provide a \textit{unified} performance analysis framework that can be applied to \textit{all} of the developed algorithms. We focus on the \textit{finite-time} analysis where, for a given stopping time $T$, the cumulative PIF loss and the convergence speed will be characterized.  In this way, we can shed important light on the fundamental differences of these algorithms, and how these differences impact their performances.

We start with the expected cumulative PIF loss defined in Sec.~\ref{sec:bandit}. For BPA and CBPA, the expected PIF loss can be written as \eqref{eqn:regdelta}. When the switching cost is considered, equation \eqref{eqn:regret} and \eqref{eqn:regdelta} should be rewritten as:
\begin{eqnarray*}
&&R^{SC}_{T} = G_T^{*}-G^S_T \\ &=&\max\limits_{i=1,..,n}\left(\sum\limits^{T}\limits_{t=1}r_{i}(t)\right)-\sum\limits^{T}\limits_{t=1}r_{a(t)}(t)+\sum\limits_{t=2}^Ts_{a(t)a(t-1)},
\end{eqnarray*}
and
\begin{eqnarray*}
\mathbb{E}[R^{SC}_{T}]&=&T\mu^*-\mathbb{E}\left(\sum\limits^T\limits_{t=1}\mu_{a(t)}-{{\sf{SC}}(T)}\right) \\
&=& \sum\limits_{i=1}\limits^n\Delta_i\mathbb{E}[N_i(T)]+\mathbb{E}[{\sf{SC}}(T)],
\end{eqnarray*}
respectively.

\subsection{Upper Bound Analysis}
In order to derive the unified framework that applies to all the algorithms, we focus on analyzing CBPA-SC as it is the most general algorithm consisting of all the key components. As we have discussed, the expected cumulative PIF loss should grow sub-linearly with $T$ in order to achieve the optimal performance, which indicates that $\lim_{T\rightarrow\infty}R_T/T=0$. We have the following theorem to bound the expected cumulative PIF loss of CBPA-SC.

\begin{theorem}
\label{thm:regretsc}
The expected cumulative PIF loss $\mathbb{E}[R^{SC}_{T}]$ of CBPA-SC is bounded above as:
\begin{equation*}
\begin{split}
&\mathbb{E}[R^{SC}_{T}] \leqslant \sum\limits_{i=1,i\neq i^*}\limits^n\Delta_i\mathbb{E}[N_i(T)] + \mathbb{E}[{\sf{SC}}(t)]
\\
& \leqslant  \sum_{i=1,i\neq i^*}^n \Big( \Delta_i(C_1^i\log T+C_2^i)+(\tilde{s}_i^{max}+\tilde{s}_{i^*}^{max})\mathbb{E}[S_i(T)]  \Big) \\
 & \qquad\qquad  + \tilde{s}_{i^*}^{max}
\\ 
&\leqslant\sum\limits_{i=1,i\neq i^*}\limits^n\Delta_i(C_1^i\log T+C_2^i)+\sum\limits_{i=1,i\neq i^*}\limits^n(\tilde{s}_i^{max}+\tilde{s}_{i^*}^{max})\quad
\\
&\Bigg(\log2C_1^i\sqrt{\log_2T} +(C_2^i+\log2C_1^i)\left(1+\frac{\pi^2}{6}\right)\Bigg)+\tilde{s}_{i^*}^{max},
\end{split}
\end{equation*}
where
\begin{equation*}
C_1^i=\frac{16\sigma_{0}^2}{\Delta_{i}^2}+\frac{\log 2}{2}\left(e^{\frac{3M_{i^*}^2}{2\sigma_0^2}}+e^{\frac{3M_{i}^2}{2\sigma_0^2}}\right),
\end{equation*}
\begin{equation*}
C_2^i=\frac{4\sigma_{0}^2}{\Delta_{i}^2}\log\sqrt{2\pi e}+\left(e^{\frac{M_{i^*}^2}{3\sigma_0^2}}+e^{\frac{M_{i}^2}{3\sigma_0^2}}\right),
\end{equation*}
$\delta_i^2=\sigma_0^2/\sigma_{i-cond}^2$, and $\sigma_{i-cond}^2=\sigma_{0}^2-{\bf{\sigma}}_{i}(0)\Sigma_{\sim i}^{-1}(0){\bf{\sigma}}_{i}^T(0)$. $M_{i}=\sigma_0^2\sqrt{1+\delta_i^2}\sum\limits_{j=1}\limits^n\sum\limits_{k=1}\limits^n|\lambda_{kj}^0||\mu_{j}^0-\mu_{j}|$ measures the accuracy of the prior knowledge, where $\Sigma_{\sim i}$ is the submatrix of $\Sigma_0$, which excludes the $i$-th column and $i$-th row and 
$\lambda^0_{kj}$ is the component of $\Lambda_0$. $\tilde{s}_i^{max}=\max_{j=1,..,n}\mathbb{E}[s_{ij}]$ is the maximum expected switching loss when SBS changes power to $p_i$.
\end{theorem}
\begin{proof}
See Appendix~\ref{apd:proof1}.
\end{proof}

Theorem~\ref{thm:regretsc} provides an $\mathcal{O}(\log T)$ upper bound for CBPA-SC, which guarantees that its cumulative PIF will converge to that of the global optimum power value at a rate of $\mathcal{O}(\log T/T)$. Furthermore, this upper bound applies to any finite time $T$ and any general function of switching loss $f(|p_i-p_j|)$ as long as $f$ is a non-decreasing and finite function.

Theorem~\ref{thm:regretsc} is a powerful result as it gives an $\mathcal{O}(\log T)$ bound for the most general algorithm CBPA-SC. We can now derive similar results for all the other proposed algorithms. First, when  $s_{ij}=0, \forall i, j \in\mathcal{N}_{pow}$, Theorem~\ref{thm:regretsc} can be applied to CBPA. Formally, we have the following corollary.

\begin{corollary}
\label{cor:cbpareg}
The expected cumulative PIF loss $\mathbb{E}[R_{T}]$ of CBPA is bounded above as:
\begin{equation*}
\mathbb{E}[R_{T}] \leq \sum\limits_{i=1,i\neq i^*}\limits^n\Delta_i\left(\lceil \frac{4\sigma_{0}^2}{\Delta_{i}^2}(\log{2\pi e}+4\log{T})-1 \rceil+\hat{N}_{i}\right),
\end{equation*}
where
\begin{equation*}
\hat{N}_{i}= e^{\frac{M_{i^*}^2}{3\sigma_0^2}}+e^{\frac{M_{i}^2}{3\sigma_0^2}}+\frac{9}{2}\left(e^{\frac{3M_{i^*}^2}{2\sigma_0^2}}+e^{\frac{3M_{i}^2}{2\sigma_0^2}}\right).
\end{equation*}
\end{corollary}
\begin{proof}
See Appendix~\ref{apd:proof2}.
\end{proof}
As Corollary~\ref{cor:cbpareg} shows, the $\mathcal{O}(\log T)$  upper bound of the cumulative PIF loss still holds for the CBPA algorithm. Thus, adding switching cost into the problem does not change the optimal scaling of the cumulative PIF loss. However, the algorithm that deals with the switching cost (CBPA-SC) is considerably more complicated than the one without the switching cost (CBPA).

Next, we note that the difference between BPA and CBPA lies in the correlation structure. We can further remove the correlation component in Corollary~\ref{cor:cbpareg} to analyze BPA.

\begin{corollary}
\label{cor:cpareg}
The expected cumulative PIF loss $\mathbb{E}[R_{T}]$ of BPA is bounded above as:
\begin{eqnarray*}
\mathbb{E}[R_{T}] &\leq& \sum\limits_{i=1,i\neq i^*}\limits^n\Delta_i\Big(\lceil \frac{4\sigma_{0}^2}{\Delta_{i}^2}(\log{2\pi e}+4\log{T})-1 \rceil \\
&&+e^{\frac{\Delta m_{i^*}^2}{3\sigma_0^2}}+e^{\frac{\Delta m_{i}^2}{3\sigma_0^2}}+\frac{9}{2}e^{\frac{3\Delta m_{i^*}^2}{2\sigma_0^2}}+\frac{9}{2}e^{\frac{3\Delta m_{i}^2}{2\sigma_0^2}}\Big),
\end{eqnarray*}
where $\Delta m_i=\mu_i-\mu_i^0$ measures the accuracy of the prior knowledge of the mean PIF.
\end{corollary}
\begin{proof}
See Appendix~\ref{apd:proof3}.
\end{proof}

Finally, because the UiPA algorithm does not use any prior knowledge, its utility function $Q_i^{UiPA}(t)$ is similar to the \textit{UCB1-NORMAL} algorithm in \cite{Auer:02}. Thus, the upper bound of the expected PIF loss can be derived analogously.

\begin{theorem}
\label{thm:uipareg}
The expected cumulative PIF loss $\mathbb{E}[R_{T}]$ of UiPA is bounded above as:
\begin{eqnarray*}
\mathbb{E}[R_{T}] &\leq& \sum\limits_{i=1,i\neq i^*}\limits^n\Delta_i\Big( \frac{16\sigma_0^2}{\Delta_i^2}(\log{2\pi e}+4\log T) \\
&&+((2\pi e)^{-1/4}+2)\log T+\frac{\log{2\pi e}}{2}+\frac{2}{\sqrt{2\pi e}} \Big),
\end{eqnarray*}
\end{theorem}
\begin{proof}
See Appendix~\ref{apd:proof4}.
\end{proof}

We can see that even though the constant terms in the upper bounds of CBPA and BPA may possibly be larger than the ones of UiPA, with a much smaller coefficient of $\log T$, the performance turns out to be better. Moreover, if the prior knowledge is accurate in BPA and CBPA, the upper bounds for both will become:
\begin{equation*}
\mathbb{E}[R_{T}] \leq \sum\limits_{\substack{i=1\\i\neq i^*}}\limits^n\Delta_i\left(\lceil \frac{4\sigma_{0}^2}{\Delta_{i}^2}(\log{2\pi e}+4\log{T})-1 \rceil+\frac{4}{\sqrt{2\pi e}}\right),
\end{equation*}
which can be easily derived from the corollaries.

\section{Reducing Complexity in Multi-SBS}
\label{sec:multi}

A practical problem in a multi-SBS deployment may arise due to the ``curse of dimensionality''. As the set of arms consists of the combinations of different power levels at all SBSs, it leads to $n^K$ arms and incurs exponential time and space complexity for the proposed algorithms. Plus, the number of available power levels for each SBS $n$ can be large. Note that in the CBPA and CBPA-SC algorithms, we need matrix calculations when updating the estimated state, which calls for $\mathcal{O}(n^{3K})$ time complexity and $\mathcal{O}(n^{2K})$ space complexity \cite{Knuth:97}. This severely limits the applicability of the proposed algorithms in large enterprise networks.

To reduce the complexity, we first explore a practical constraint that has not been utilized in the proposed algorithms. In real-world deployment, the neighboring SBSs are generally not allowed to have vastly different pilot power levels. This is because otherwise they may result in significantly different coverage areas and therefore lead to very uneven load distributions. Thus, utilizing this practical constraint, we  only need to consider the combinations of power levels in which neighboring SBS power levels are different by no more than a certain threshold $P_{th}$.

Even with the power difference threshold, the size of set is still exponential in $K$. To further reduce the complexity, we notice that the performance space of all set of arms exhibits certain ``clustering'' effect that can be utilized. For two power settings that differ only slightly (e.g., $\{0,3,5\}$ and $\{0,4,4\}$ dBm for $K=3$), the performances may be very similar. Thus, if we can carefully group the power settings into a few clusters, and only use the cluster center as the representative power setting, we can achieve a good tradeoff between complexity and performance for the  algorithms.

We propose to perform a clustering operation to address the complexity issue. The clustering operation is done after the self-configuration phase to leverage the prior knowledge, but before invoking the power assignment algorithm. We adopt the \textit{K-medoids} clustering \cite{Park:09} because, different from the well-known K-means clustering, K-medoids is based on the most central object instead of the centroids in K-means, each of which is the mean point of all objects in the cluster. Therefore, the medoids in each cluster can be seen as the representative power settings. We note that the choice of the number of clusters $N$ plays a critical role in the overall performance. If it is too large, the global optimum power setting may be a medoid with high probability, which contributes to high accuracy for the power assignment process but also increases the complexity and leads to low efficiency, and vice versa.

We further note that there is a $\mathcal{O}(n^KN)$ time complexity for K-medoids clustering \cite{Park:09}, but as clustering is done prior to the self-optimization phase, the process can be handled offline. Thus, time complexity is less of a concern.

\section{Simulation Results}
\label{sec:sim}

\begin{figure*}[t]
    \centerline{
        \subfigure[ MBS-SBS distance is 150m. ] {
        { \includegraphics[width=0.49\textwidth]{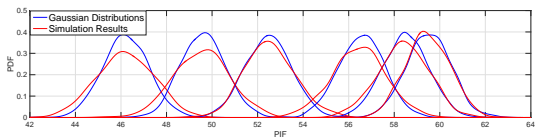} \label{simugaus2} } }
        \hfil
        \subfigure[ MBS-SBS distance is 70m. ] {
        { \includegraphics[width=0.42\textwidth]{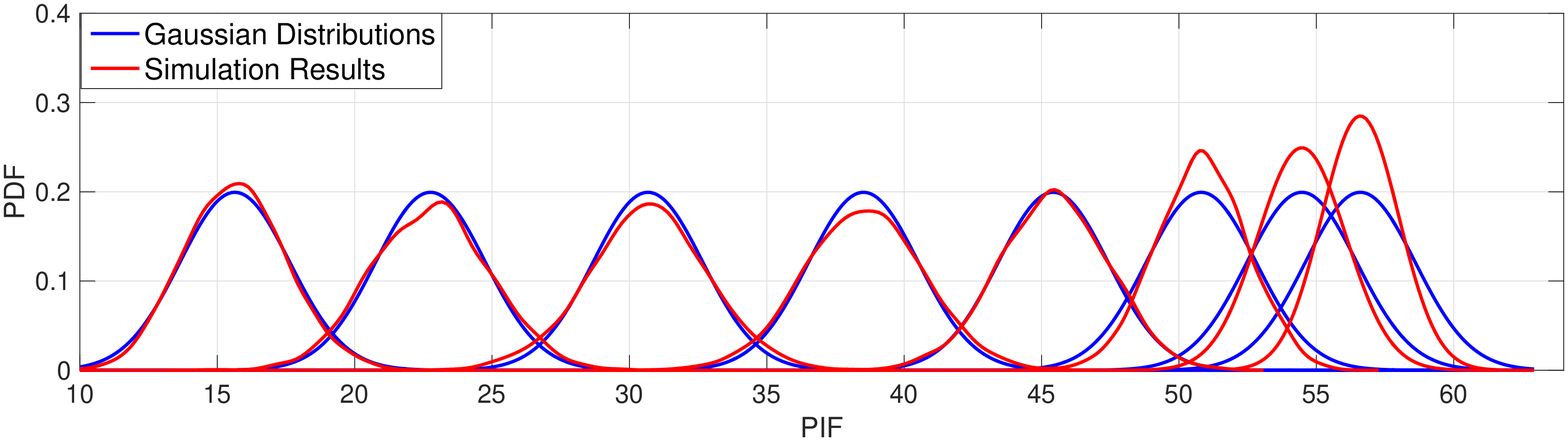} \label{simugaus3} }  }
        }
    \caption{A $40\times 30 m^2$ warehouse, with elliptic routes inside and outside.}
    \label{simugaus23}
\end{figure*}

\subsection{Simulation setup}

We resort to numerical simulations to verify the effectiveness of the developed transmit power assignment algorithms. A system-level heterogeneous network simulator is developed considering both indoor SBS and outdoor MBS. We consider a large warehouse with $K=1,2,4$ SBSs deployed inside and a MBS outside with a fixed transmit power setting. The measurement points constitute two routes inside and outside respectively which we assume to follow concentric circle or ellipse pattern. 100 measurement points are set uniformly on each route. At each time slot, the measurement points feedback their own coverage condition, determined by the respective SINR which is naturally decided by the current SBS power setting. We set the total time horizon as $T=3000$ slots and iterate each simulation setting for 50 times to average out the randomness. The size of the warehouse and the SBS locations are given in Table~\ref{tab:parameter}. Here we set the center of the warehouse as origin. The PIF $r$ under each power value can be calculated following the procedure in Sec.~\ref{sec:sys}.  

We obtain the received power from SBS or MBS using the indoor femto channel model of urban deployment from \cite{Pathloss} as follows.
\begin{itemize}
\item indoor UE to MBS:
\begin{equation}
\label{indoorMBS}
PL(d)[\text{dB}]=15.3+37.6\times\log_{10}(d)+L_{ow}+X_{\sigma_{dB}}, 
\end{equation}
\item outdoor UE to MBS:
\begin{equation}
\label{outdoorMBS}
PL(d)[\text{dB}]=15.3+37.6\times\log_{10}(d)+X_{\sigma_{dB}},
\end{equation}
\item indoor UE to SBS:
\begin{equation}
\label{indoorSBS}
PL(d)[\text{dB}]=38.46+20\times\log_{10}(d)+X_{\sigma'_{dB}}, 
\end{equation}
\item outdoor UE to SBS:
\begin{equation}
\label{outdoorSBS}
\begin{aligned}
PL(d)[&  \text{dB}]= \max\{  15.3+37.6\log_{10}(d),  \\
& 38.46+20\log_{10}(d)\}+L_{ow}+X_{\sigma'_{dB}}. 
\end{aligned}
\end{equation}
\end{itemize}
Note that \eqref{indoorMBS} and \eqref{indoorSBS} are for indoor routes while \eqref{outdoorMBS} and \eqref{outdoorSBS} are for outdoor routes; $d$ represents the separation between a BS and the measurement point; $L_{ow}$ is the penetration loss of an outdoor wall, which indoor user suffers when receiving power from outdoor MBS and outdoor user receiving from indoor SBS; $X_{\sigma_{dB}}$ and $X_{\sigma'_{dB}}$ stand for shadow fading. Other important simulation parameters are summarized in Table~\ref{tab:parameter}. 

\begin{table}[h]
\caption{Simulation Parameters}
\label{tab:parameter}
\begin{center}
\begin{tabular}{| l | l |}
\hline
\textbf{Parameters}      &  \textbf{Value} \\
\hline SBS transmit power &  [-10dBm, 20dBm] \\ 
\hline MBS transmit power   &  40dBm\\
\hline Thermal noise density    & -174dBm/Hz  \\
\hline Bandwidth      & 20MHz  \\
\hline Carrier frequency      &  2GHz\\
\hline Penetration loss ($L_{ow}$)      &  20dB \\ 
\hline Shadowing effect & \tabincell{l}{log-normal with \\ $\sigma=8\text{dB}$, $\sigma'=4\text{dB}$} \\
\hline  $d_0$ & 1m \\
\hline $\alpha$ & 0.7 \\
\hline
\multirow{3}{*}{}{Enterprise Size}& K=1 30m$\times$30m \\
 & K=2 40m$\times$40m   \\
 & K=4 50m$\times$40m  \\
\hline
\multirow{4}{*}{}{SBS location} & K=1 (12m,8m) \\
 &  K=2 (16m,17m), (-15m,-11m) \\
 & K=4  (20m,18m), (11m,-19m)\\
 & \quad (-11m,18.5m), (-10.5m,-19m)  \\
\hline
\end{tabular}
\end{center}
\end{table}

\subsection{Evaluation of the PIF Gaussian Distribution}

We first study the empirical distribution of the PIF $r$ in $K=1$ SBS with the set of power levels $\mathcal{P}=\{-10,-5,.., 15, 20\}$ dBm. We present the comparison of empirical and Gaussian distributions in two representative scenarios in  Fig.~\ref{simugaus2} and \ref{simugaus3}. As we can see, the assumption on Gaussian distributed PIFs matches well with the empirical distributions.

To further verify the dependency on the Gaussian distribution, we study the performance of the proposed algorithms compared with a well-behaved UCB extended algorithm which makes no assumptions on the distribution of the rewards, e.g. UCB-V in \cite{Audibert:08} under non-Gaussian reward distributions. More specifically, two well-adopted distributions in wireless communications, \textit{uniform} and \textit{Rayleigh}, are considered.  We can see from Fig.~\ref{nonGau} that performances under non-Gaussian distributions are still very good, particularly for BPA and CBPA. This observation indicates that Gaussianness is not a fundamental assumption that must be met to guarantee the effectiveness of the algorithms. 


\begin{figure}[h]
        \begin{center}
        \subfigure[ Uniform ]{
        \includegraphics[width=0.48\textwidth]{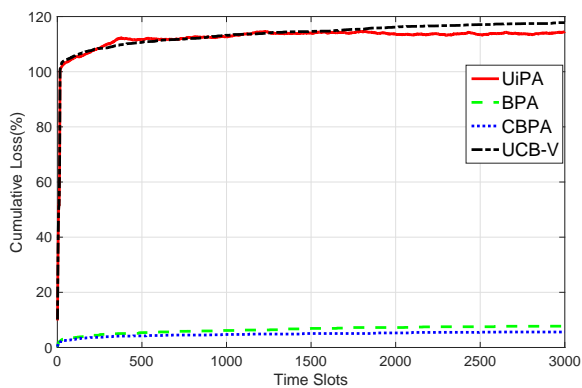} } 
        \end{center}
        \hfil
        \subfigure[ Rayleigh ]{
        \includegraphics[width=0.48\textwidth]{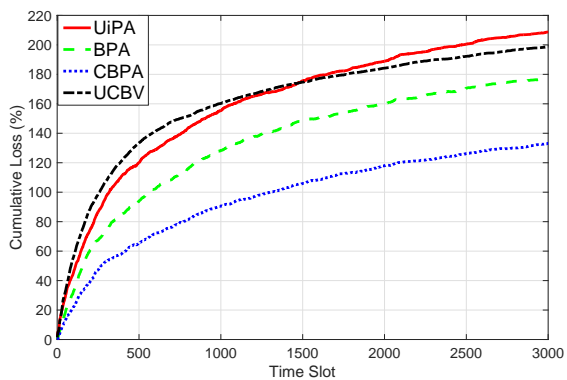}
        }
    \caption{Verifications of algorithms under non-Gaussian distributed rewards.}
    \label{nonGau}
\end{figure}

\subsection{System Performance}

In the simulation setting for $K=1$, we deploy an outside MBS at $[100m,100m]$. The set of power levels for SBS is $\mathcal{P}=\{-10,-8,..,18,20\}$ dBm while other settings follow Table~\ref{tab:parameter}. The inside and outside routes have the concentric circle pattern, whose radiuses are (2, 13) meters for the two indoor routes, and (24, 30) meters for the two outdoor routes. The cumulative loss over time is used to evaluate the performance, and we use the optimal power achieved by the global optimization of the expected PIF with complete RF information as the genie-aided optimum.


\begin{figure}[h]
        \begin{center}
        \subfigure[ Good quality (with estimated mean) ]{
        \includegraphics[width=0.48\textwidth]{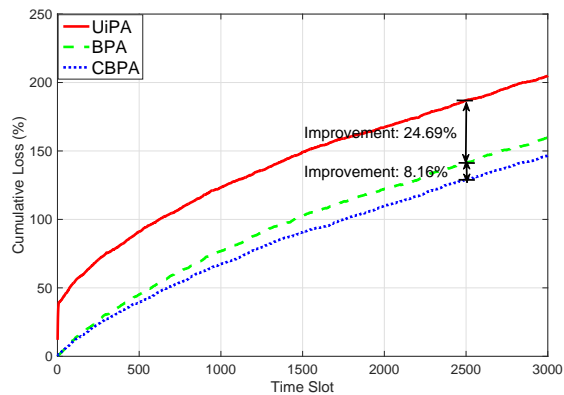} } 
        \end{center}
        \hfil
        \subfigure[ Poor quality (with uniform distribution) ]{
        \includegraphics[width=0.48\textwidth]{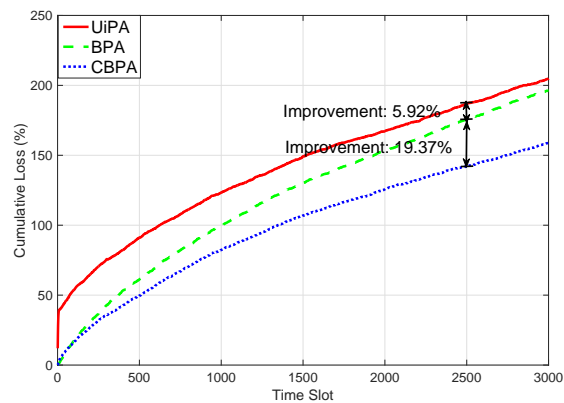}
        }
    \caption{Cumulative loss comparison of prior knowledge with different qualities in a single-SBS deployment with $\alpha=0.7$.}
    \label{wellpoor}
\end{figure}

We first compare the performance of UiPA, BPA and CBPA algorithms with different quality of priors. Fig.~\ref{wellpoor}(a) reports the cumulative loss over time for all three algorithms when the prior knowledge is of good quality, i.e. the estimated mean from the empirical distribution is used. Fig.~\ref{wellpoor}(b) shows the same simulation but with a poor prior knowledge, which uses a uniform prior distribution with each element $\mu_{0}=50$. A few important observations can be made from these simulations. First of all, we see that all three algorithms can converge to the optimal power value asymptotically, but with different speed. To further evaluate the convergence speed, we plot the empirical CDF of the convergence time for all three algorithms in Fig.~\ref{convergecdf}. It becomes clear that leveraging both the prior knowledge and the correlation structure significantly accelerates the convergence of CBPA. In terms of minimizing the total PIF loss, CBPA also outperforms BPA which performs better than UiPA. Second, degradation of the quality of the priors degrades the performance of {BPA} and {CBPA}. Particularly, performance of the {BPA} is getting close to {UiPA} with poor prior knowledge. It is interesting to note that even with poor prior, CBPA still converge faster than other algorithms with good prior. This is because when the prior knowledge is inaccurate, CBPA recovers some of the PIF degradation by leveraging its correlation structure. Lastly, as UiPA does not leverage the prior knowledge, changing its quality does not affect the convergence speed.

\begin{figure}[h]
\centerline{\includegraphics[width=0.48\textwidth]{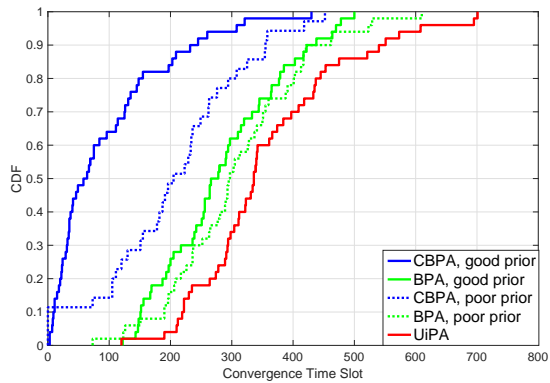}}
\caption{Convergence speed comparison of prior knowledge with different qualities in a single-SBS deployment with $\alpha=0.7$. ``Good prior'' corresponds to using the estimated mean while ``poor prior'' uses a uniform distribution.}
\label{convergecdf}
\end{figure}

Next, we compare the proposed algorithms with the industry solution. The heuristic solution \cite{Sumeeth} keeps a power value long enough to obtain a near-perfect PIF estimation, and then it either increases or decreases the power value by a fixed step size. Clearly, this method trades off fast convergence for certainty. Fig.~\ref{industry} reports the numerical comparison with a maximum $20$dBm and step size $2$dB. We can see that the industrial solution adapts poorly to different deployments, while our algorithms are stable thanks to  online learning. 


\begin{figure}[h]
        \begin{center}
        \subfigure[ Large warehouse ]{
        \includegraphics[width=0.48\textwidth]{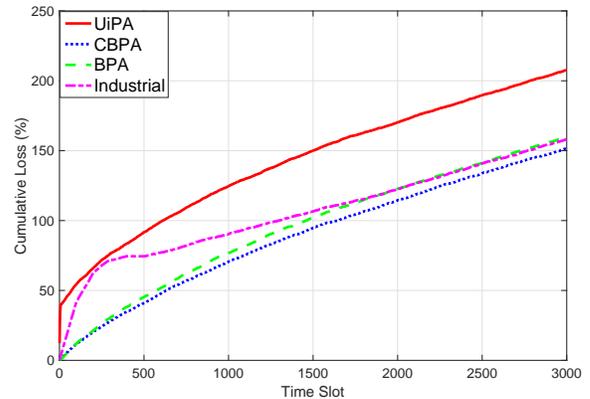} } 
        \end{center}
        \hfil
        \subfigure[ Small single-office enterprise ]{
        \includegraphics[width=0.48\textwidth]{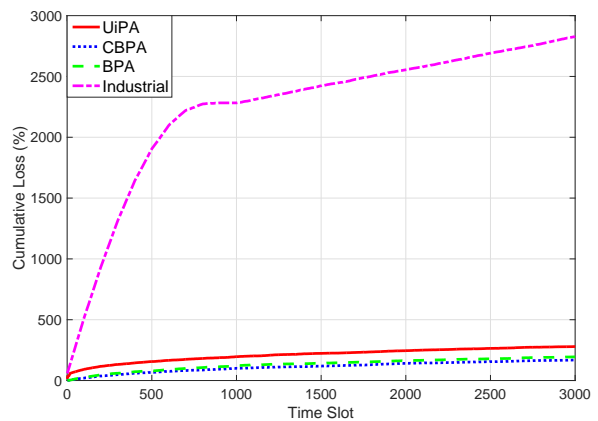} }
    \caption{Comparison of the industrial solution to UiPA, BPA and CBPA in different scenarios.}
    \label{industry}
\end{figure}

For $K=2$ and $K=4$, an outside MBS is deployed at $[100m,100m]$. The power value difference threshold is $P_{th}=5$dB. The power value for each SBS is selected from $\mathcal{P}=\{-10,-5,..,15,20\}$dBm. It results in $n=19$ for $K=2$ without any clustering, which may be acceptable in terms of complexity. The cumulative PIF loss with respect to the optimal power setting is shown in Fig.~\ref{douqua}(a). For $K=4$ case, however, there are $n=149$ power settings. We thus employ the clustering strategy in Sec.~\ref{sec:multi} and study two cases where the number of clusters is either $N=20$ or $N=40$. The PIF loss normalized by time is shown in Fig.~\ref{douqua}(b). We can see that all algorithms exhibit a decaying loss per slot. As for the effect of $N$, there exists an initial period when larger cluster number results in worse performance for all the algorithms. This is because during the initial slots, more power settings lead to more exploration and thus sub-optimal power settings are selected more. As time goes by, the algorithms have more knowledge about the optimal power setting. While a larger cluster number means one of the selected clustering medoids  is closer to the globally optimal power setting, a larger $N$ results in a better performance. Detailed coverage and leakage results under optimal selections are reported in Table \ref{tab:multi}.

\begin{table*}[ht]
\caption{Multi-SBS simulation results}
\label{tab:multi}
\begin{center}
\begin{tabular}{| l || c | c | c|}
\hline
\textbf{Metric}       &  \textbf{K=2}       &\textbf{K=4} \\
\hline
Globally optimal power [dBm]  & (0, 5)    & (0, 5, 10, 15)  \\
\hline
Coverage  percentage      &  91.506\%   & 96.548\% \\
\hline
Leakage  percentage      &  5.691\%   &  28.725\% \\
\hline
Simulation output power [dBm] &  (0, 5)   &   (-5, 0, 5, 10) when $N=20$ \\
&&  (0, 5, 10, 15) when $N=40$ \\
\hline
\end{tabular}
\end{center}
\end{table*}


\begin{figure}[h]
        \begin{center}
        \subfigure[ Two SBSs deployed in a $40\text{m} \times 40\text{m}$ enterprise ]{
        \includegraphics[width=0.48\textwidth]{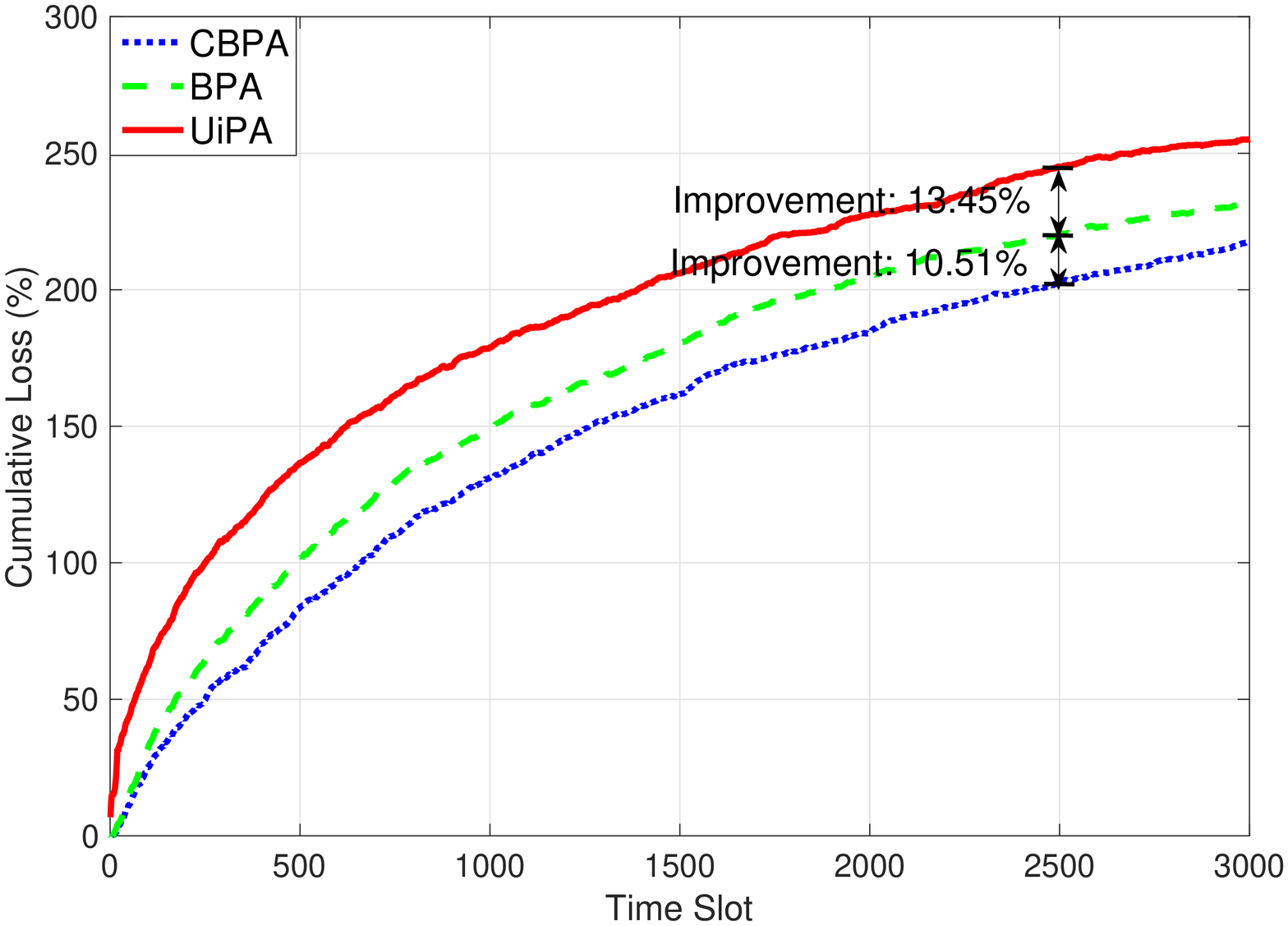} } 
        \end{center}
        \hfil
        \subfigure[ Four SBSs deployed in a $50\text{m} \times 40\text{m}$ enterprise ]{
        \includegraphics[width=0.48\textwidth]{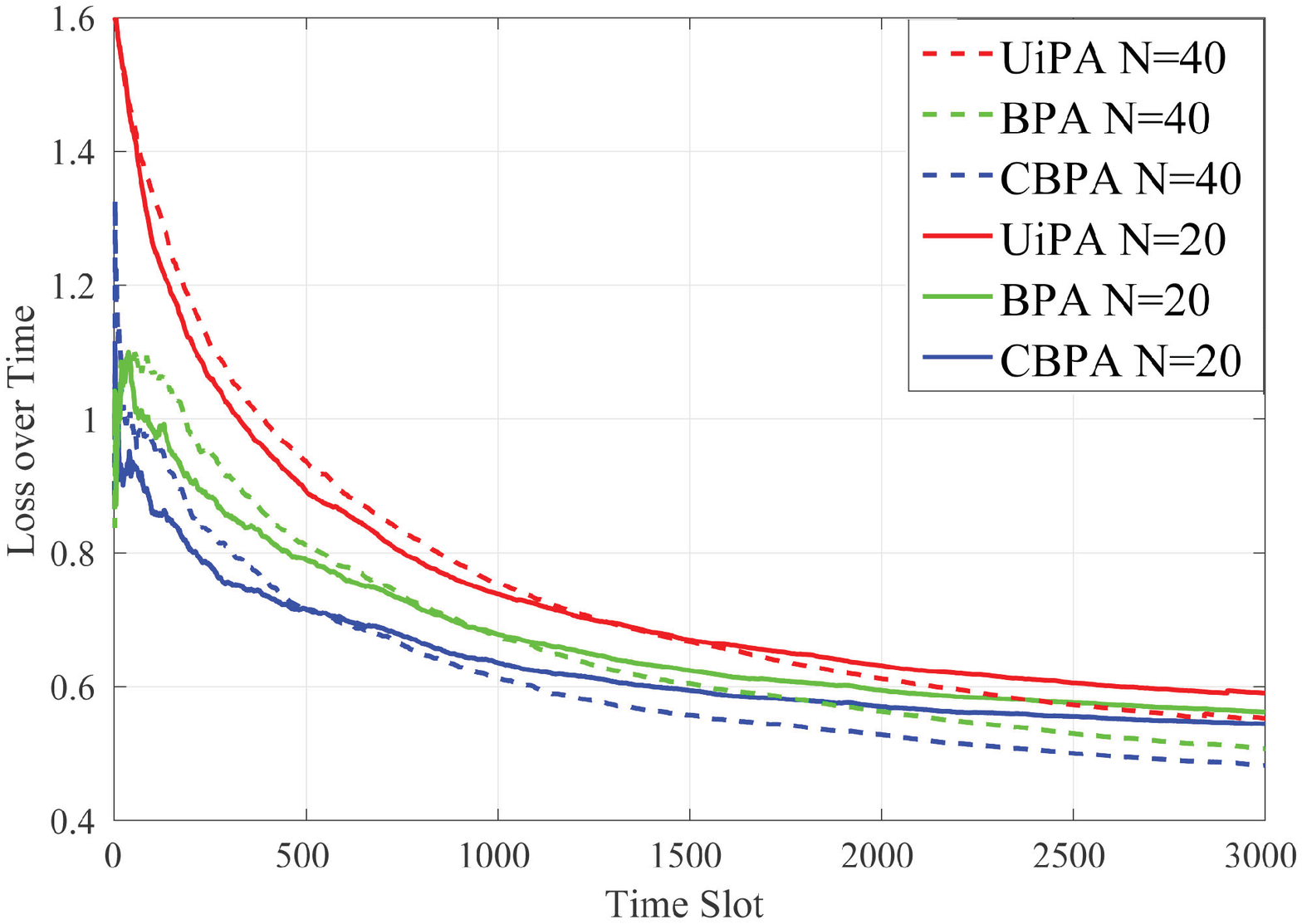} }
    \caption{Simulation results for two and four SBSs in different scenarios.}
    \label{douqua}
\end{figure}


Fig.~\ref{blockdou} and \ref{blockqua} study the impact of power switching cost for $K=2$ and $K=4$, respectively. Here we adopt a simple linear function of switching loss as $s_{ij}=\gamma|p_i-p_j|$, where $\gamma$ is a tunable parameter for different scenarios and we set as 0.2. We can see that the additional performance loss occurring whenever a SBS changes its power value increases the overall performance loss in all algorithms. However, the algorithms can still converge to the optimal power settings asymptotically in a sub-linear fashion, matching the regret analysis in Sec.~\ref{sec:alg}. In Fig.~\ref{blockqua}, the performances of different cluster numbers also comply with our previous analysis.


\begin{figure}[h]
        \begin{center}
        \subfigure[ Cumulative loss, $K=2$ ]{
        \includegraphics[width=0.48\textwidth]{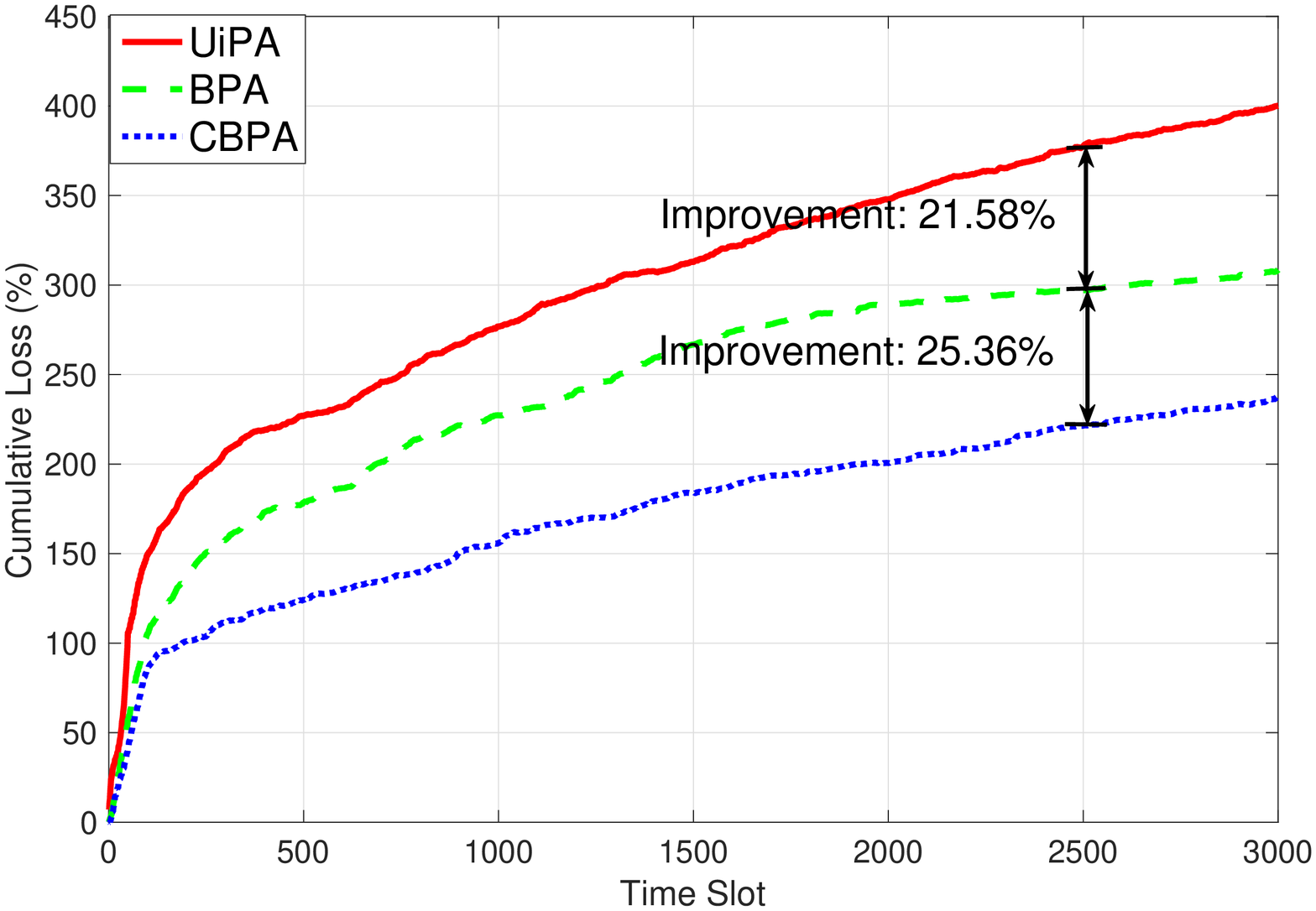} \label{blockdou} } 
        \end{center}
        \hfil
        \subfigure[ Per-slot loss, $K=4$ ]{
        \includegraphics[width=0.48\textwidth]{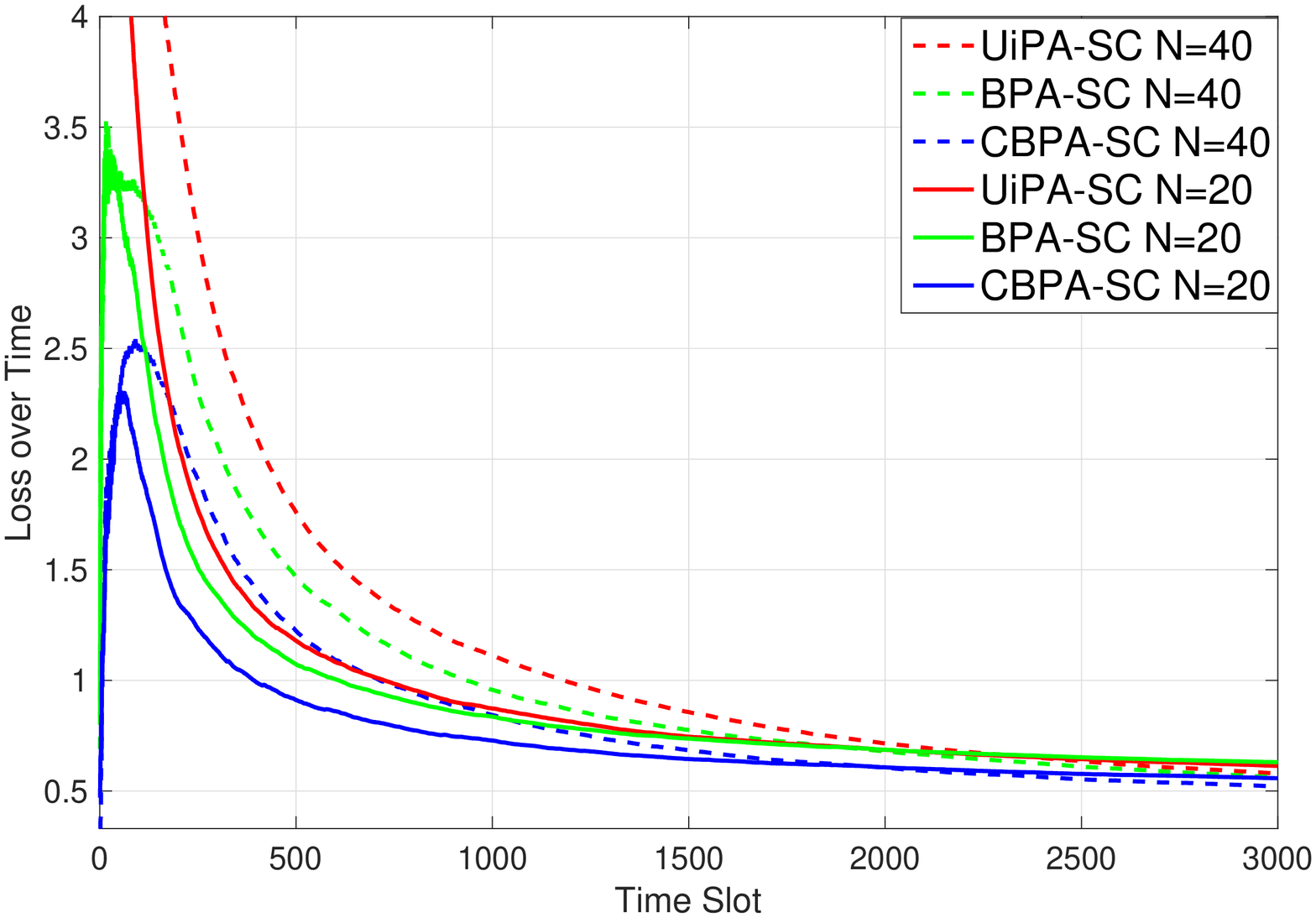} \label{blockqua}}
    \caption{Performance with switching cost factor $\gamma=0.2$.}
    \label{blockdouqua}
\end{figure}

\section{Conclusion}
\label{sec:conc}

We have studied the pilot power assignment problem associated with indoor enterprise closed-access SBS networks, in which the focus is on achieving optimal balance between providing sufficient coverage for the indoor users and suppressing leakage that causes interference to outdoor MBS users. We modeled power assignment as an online learning problem, and  adopted a Bayesian approach that leverages the prior information of the Gaussian distribution. We proposed bandit-inspired power assignment algorithms that utilize different levels of the statistical information. The CBPA algorithm makes use of both prior knowledge of the mean and variance of each arm as well as the dependency of PIFs across different power values. In contrast, the BPA algorithm only uses the prior knowledge but not the correlation information, and its performance is worse than CBPA but better than the UiPA algorithm that does not use either prior or correlation. Furthermore, we explicitly took into account the power switching cost, and enhanced the power assignment algorithms with a block allocation scheme to reduce frequent power-switchings. A sub-linear upper bound for performance loss was proved for all the algorithms. Furthermore, for the multi-SBS deployment, we proposed to use K-medoids clustering to reduce the complexity while maintaining the performance. When the cluster number becomes large, the algorithms can approach the globally optimal power setting for all $K$ SBSs. 

As a possible future direction, the \textit{spectral bandits} method proposed in \cite{Valko:14} offers a new perspective to efficiently handle a large number of arms while capturing the correlation structure. This can be an interesting alternative for the enterprise transmit power assignment problem. In particular, complexity and performance comparison with the algorithms of this paper may shed light into its feasibility.

\appendices

\section{Proof of Theorem \ref{thm:regretsc}}
\label{apd:proof1}

We start by proving for the case $L_{l-1}<T\leqslant L_l$. Note that
\begin{eqnarray}
\nonumber && N_i(T)=\sum\limits_{t=1}\limits^T\mathcal{I}(p_{a(t)}=i) \leqslant \sum\limits_{t=1}\limits^T\mathcal{I}(Q_i^t \geqslant Q_{i^*}^t)
\\     &\leqslant& \eta_i+\sum\limits_{t=1}\limits^{T}\mathcal{I}(Q_{i}^t \geqslant Q_{i^*}^t,N_{i}(t-1)\geqslant\eta_{i})\label{eqn:NiT}
\\    &\leqslant& \eta_i+\sum\limits_{f=1}\limits^ l\sum\limits_{k=1}\limits^{b_f}f \mathcal{I}(Q_{i}^t \geqslant Q_{i^*}^t,N_{i}(\tau_{fk})\geqslant\eta_{i}),\label{equ:nit}
\end{eqnarray}
where $i^*=\arg\max_{i=1,..,n}\mu_{i}$, $\eta_{i}$ is a positive integer, and $\mathcal{I}(x)$ is the indicator function. At any time $t$, sub-optimal $i$ is selected only when $Q_{i^*}^t \leqslant Q_{i}^t$, which is true as long as one of the following inequalities holds:
\begin{subequations}
\begin{equation}
\hat{\mu}_{i^*}(\tau_{fk})\leqslant \mu_{i^*}-U_{i^*}(\tau_{fk}), \label{equ:1}
\end{equation}
\begin{equation}
\hat{\mu}_{i}(\tau_{fk})\geqslant \mu_{i}+U_{i}(\tau_{fk}), \label{equ:2}
\end{equation}
\begin{equation}
\mu_{i^*}<\mu_{i}+2U_{i}(\tau_{fk}), \label{equ:3}
\end{equation}
\end{subequations}
where $U_{i}(\tau_{fk})=\hat{\sigma}_{i}(\tau_{fk})\sqrt{\sum\limits_{j=1}\limits^n{\rho_{ij}^2(\tau_{fk})}}\Phi^{-1}(1-1/\sqrt{2\pi e}\tau_{fk}^2)$.
Define the bias $\bm{e}$ and covariance $\bar{\Sigma}$ of the estimate $\hat{\bm{\mu}}(t)$, with $e_{i}$ and $\bar{\sigma}_{i}$ representing the $i$-th entry of $\bm{e}$ and the diagonal of $\bar{\Sigma}$, and we have $\hat{\bm{\mu}}(t)\sim \mathcal{N}(\bm{e}(t)+\bm{\mu},\bar{\Sigma}(t))$, with $e_{i}(t)=\sum_{j=1}^n\sum_{k=1}^n\hat{\sigma}_{ik}(t)\lambda_{kj}^0(\mu^0_j-\mu_j)$.

We now separately analyze \eqref{equ:1}, \eqref{equ:2}, and \eqref{equ:3}. First,  if $N_{i^*}(\tau_{fk})=0$, then \eqref{equ:1} is false if \cite[Lemma 7]{Srivastava:15}
\begin{equation*}
U_{i^*}(\tau_{fk})> \sigma_{i*-cond}\sqrt{3\log{\tau_{fk}}}\geqslant \frac{M_{i^*}}{\sqrt{1+\delta_{i^*}^2}}\geqslant|e_{i^*}(\tau_{fk})|
\end{equation*}
or equivalently, 
\begin{equation}
\label{equ:tau1}
\tau_{fk}>e^{\frac{M_{i^*}^2\delta_{i^*}^2}{3\sigma_0^2(1+\delta_{i^*}^2)}}.
\end{equation}
Otherwise, if $N_{i^*}(\tau_{fk})\geqslant1$, we have 
\begin{eqnarray}
&& \mathbb{P}\{\hat{\mu}_{i^*}(\tau_{fk}) \leqslant \mu_{i^*}-U_{i^*}(\tau_{fk})\}
\nonumber \\ 
& \leqslant & \mathbb{P}\left\{z\geqslant\Phi^{-1}(1-1/\sqrt{2\pi e}\tau_{fk}^2)-\frac{M_{i^*}}{\sigma_0}\right\} \nonumber  \\
& \leqslant & \mathbb{P}\left\{z\geqslant\sqrt{3\log{\tau_{fk}}}-\frac{M_{i^*}}{\sigma_0}\right\}, \label{eqn:P}
\end{eqnarray}
where $z$ is a standard Gaussian random variable. This  indicates that $\sqrt{3\log{\tau_{fk}}}-\frac{M_{i^*}}{\sigma_0}\geqslant0$. Thus we have $\tau_{fk}>e^{M_{i^*}^2/3\sigma_0^2}=\tau_1$. For $\tau_{fk}>\tau_1$, we have
\begin{eqnarray}
\nonumber\mathbb{P}\{\eqref{equ:1} \text{ holds}\}&\leqslant&\frac{1}{2}exp\left(-\frac{1}{2}\left( \sqrt{3\log{\tau_{fk}}}-\frac{M_{i^*}}{\sigma_0}\right)^2\right)
\\ &\leqslant&\frac{1}{2}exp\left(-\frac{1}{2}\left(\frac{9}{4}\log{\tau_{fk}}-3\frac{M_{i^*}^2}{\sigma_0^2}\right)\right) \nonumber \\
&=&\frac{1}{2}e^{\frac{3M_{i^*}^2}{2\sigma_0^2}}\tau_{fk}^{-\frac{9}{8}}.\label{equ:1pro} 
\end{eqnarray}
Inequality \eqref{equ:1pro} is deduced using \cite[Lemma 2]{Srivastava:15}.

Similarly, we can deduce that if $N_i(\tau_{fk})>\eta_i$ and $\tau_{fk}\geqslant \tau_2:=e^{M_i^2/3\sigma_0^2}$, then
\begin{eqnarray}
\label{equ:2pro}
\mathbb{P}\{\eqref{equ:2} \text{ holds}\}\leqslant\frac{1}{2}e^{\frac{3M_i^2}{2\sigma_0^2}}\tau_{fk}^{-\frac{9}{8}}.
\end{eqnarray}
For inequality \eqref{equ:3}, it holds if 
\begin{eqnarray}
\nonumber && \mu_{i^*}-\mu_i<2U_{i}(\tau_{fk}) \\
\nonumber&&  \Longrightarrow\Delta_i<\frac{2\sigma_{0}}{\sqrt{1+N_{i}(\tau_{fk})}}\Phi^{-1}(1-1/\sqrt{2\pi e}\tau_{fk}^2)
\\ &&  \nonumber \Longrightarrow N_i(\tau_{fk}) < \frac{4\sigma_{0}^2}{\Delta_{i}^2}(\log{2\pi e}+4\log{T})-1.
\end{eqnarray}
Thus we have that \eqref{equ:3} does not hold if
\begin{equation}
\label{eqn:Nitau}
N_{i}(\tau_{fk})\geqslant \frac{4\sigma_{0}^2}{\Delta_{i}^2}(\log{2\pi e}+4\log{T})-1.
\end{equation}
Setting $\eta_i=\lceil \frac{4\sigma_{0}^2}{\Delta_{i}^2}(\log{2\pi e}+4\log{T})-1 \rceil$ and combining \eqref{equ:tau1}, \eqref{equ:1pro} and \eqref{equ:2pro}, the inequality \eqref{equ:nit} can be written as
\begin{equation}
\begin{aligned}
\label{equ:reg}
 \mathbb{E}[N_i(T)]\leqslant \eta_i+\tau_1+\tau_2+\frac{1}{2}\left(e^{\frac{3M_{i^*}^2}{2\sigma_0^2}}+e^{\frac{3M_{i}^2}{2\sigma_0^2}}\right)\sum_{f=1}^l\sum_{k=1}^{b_f}f\tau_{fk}^{-\frac{9}{8}}.
\end{aligned}
\end{equation}
We now focus on $\sum_{f=1}^l\sum_{k=1}^{b_f}f\tau_{fk}^{-\frac{9}{8}}$. With $\tau_{fk}=L_{f-1}+1+(k-1)f$ and $2^{f^2}\leqslant L_{f} \leqslant 2^{f^2}+f^2$, we have
\begin{equation*}
\begin{aligned}
& \sum_{k=1}^{b_f}f\tau_{fk}^{-\frac{9}{8}} \leqslant   \sum_{k=1}^{b_f}f(2^{(f-1)^2}+(f-1)^2+1+(r-1)f)^{-9/8} \\ 
& \hspace{4em} \leqslant  \sum_{k=1}^{b_f}\frac{f}{2^{(f-1)^2}+(f-1)^2+1+(r-1)f}  \\ 
&\hspace{4em} \leqslant \int^{b_f}_1\frac{f}{2^{(f-1)^2}+(f-1)^2+1+(r-1)f}dr \\ 
&\hspace{4em} = \log\frac{2^{f^2}+(f-1)^2+1}{2^{(f-1)^2}+(f-1)^2+1}  \\ 
&\hspace{4em} \leqslant  \log\frac{2^{f^2}}{2^{(f-1)^2}},
\end{aligned}
\end{equation*}
and
\begin{equation*}
\sum_{f=1}^l\sum_{k=1}^{b_f}f\tau_{fk}^{-\frac{9}{8}} \leqslant \sum_{f=1}^l\log\frac{2^{f^2}}{2^{(f-1)^2}}=l^2\log2\leqslant \log 2T.
\end{equation*}
Therefore \eqref{equ:reg} yields
\begin{eqnarray*}
\mathbb{E}[N_i(T)]&\leqslant& \eta_i+\tau_1+\tau_2+\frac{1}{2}\left(e^{\frac{3M_{i^*}^2}{2\sigma_0^2}}+e^{\frac{3M_{i}^2}{2\sigma_0^2}}\right)\log 2T  \\
&\leqslant& C_1^i\log T+C_2^i,  \\ 
C_1^i&=&\frac{16\sigma_{0}^2}{\Delta_{i}^2}+\frac{\log 2}{2}\left(e^{\frac{3M_{i^*}^2}{2\sigma_0^2}}+e^{\frac{3M_{i}^2}{2\sigma_0^2}}\right),  \\
C_2^i&=&\frac{4\sigma_{0}^2}{\Delta_{i}^2}\log\sqrt{2\pi e}+\left(e^{\frac{M_{i^*}^2}{3\sigma_0^2}}+e^{\frac{M_{i}^2}{3\sigma_0^2}}\right).
\end{eqnarray*}

We then establish the expected number of switches to a sub-optimal arm $i$ from a different arm. We have
\begin{eqnarray*}
 S_i(T)&\leqslant& 1+\sum\limits_{f=1}\limits^l\frac{N_i(L_f)-N_i(L_{f-1})}{f} \\
&=&1+\sum\limits_{f=1}\limits^l\frac{N_i(L_{f})}{f}-\sum\limits_{f=0}\limits^{l-1}\frac{N_i(L_{f-1})}{f+1} 
\\ &=&\frac{N_i(L_l)}{l}+\sum\limits_{f=1}\limits^{l-1}N_i(L_f)\left(\frac{1}{f}-\frac{1}{f+1}\right) \\ 
 &\leqslant& \frac{N_i(L_l)}{l}+\sum\limits_{f=1}\limits^{l-1}\frac{1}{f^2},
\end{eqnarray*}
using the same argument as \cite{Agrawal:88}. Then it follows that 
\begin{equation}
\label{equ:sit}
\mathbb{E}[S_i(T)]\leqslant \frac{\mathbb{E}[N_i(L_l)]}{l}+\sum\limits_{f=1}\limits^{l-1}\frac{\mathbb{E}[N_i(L_f)]}{f^2}.
\end{equation}
With the upper bound on $\mathbb{E}[N_i(T)]$ and $L_f\leqslant 2^{f^2}+f^2 \leqslant 2^{f^2+1}$, \eqref{equ:sit} can be further deduced as
\begin{eqnarray}
\nonumber && \mathbb{E}[S_i(T)] \leqslant \frac{C_1^i\log L_l+C_2^i}{l}+\sum\limits_{f=1}\limits^{l-1}\frac{C_1^i\log L_f+C_2^i}{f^2}
\\\nonumber &\leqslant& \frac{C_2^i}{l}+\sum\limits_{f=1}\limits^{l-1}\frac{C_2^i}{f^2}+\frac{C_1^i\log2^{l^2+1}}{l}+\sum\limits_{f=1}\limits^{l-1}\frac{C_1^i\log{2^{f^2+1}}}{f^2}
\\\nonumber&\leqslant& C_2^i\left(1+\frac{\pi^2}{6}\right)+\log2C_1^i\left(l+\frac{\pi^2}{6}\right)
\\ \nonumber &\leqslant& \log2C_1^i\sqrt{\log_2T}+(C_2^i+\log2C_1^i)\left(1+\frac{\pi^2}{6}\right).
\end{eqnarray}
Finally, the cumulative switching cost can be bounded as
\begin{eqnarray}
{{\sf{SC}}(T)} &\leqslant& \sum\limits_{i=1,i\neq i^*}\limits^n\tilde{s}_i^{max}\mathbb{E}[S_i(T)]+\tilde{s}_{i^*}^{max}\mathbb{E}[S_{i^*}(T)] \nonumber \\
& \leqslant&  \sum\limits_{i=1,i\neq i^*}\limits^n(\tilde{s}_i^{max}+\tilde{s}_{i^*}^{max})\mathbb{E}[S_i(T)]+\tilde{s}_{i^*}^{max}. \nonumber
\end{eqnarray}

\section{Proof of Corollary \ref{cor:cbpareg}}
\label{apd:proof2}

For the CBPA algorithm, \eqref{eqn:NiT} still holds. Hence, the argument from \eqref{equ:1} to \eqref{eqn:Nitau} equally applies to any time slot $t=1,2,..,T$. The proof is complete by rewriting \eqref{eqn:NiT}  as
\begin{eqnarray}
\nonumber\mathbb{E}[N_i(T)]&\leqslant& \eta_i+\tau_1+\tau_2+\frac{1}{2}\left(e^{\frac{3M_{i^*}^2}{2\sigma_0^2}}+e^{\frac{3M_{i}^2}{2\sigma_0^2}}\right)\sum\limits_{t=1}\limits^Tt^{-\frac{9}{8}}
\\ \nonumber &\leqslant&\lceil \frac{4\sigma_{0}^2}{\Delta_{i}^2}(\log{2\pi e}+4\log{T})-1 \rceil+\hat{N}_{i},
\\ \nonumber  \hat{N}_{i}&=&e^{\frac{M_{i^*}^2}{3\sigma_0^2}}+e^{\frac{M_{i}^2}{3\sigma_0^2}}+\frac{9}{2}\left(e^{\frac{3M_{i}^2}{2\sigma_0^2}}+e^{\frac{3M_{i}^2}{2\sigma_0^2}}\right).
\end{eqnarray}

\section{Proof of Corollary \ref{cor:cpareg}}
\label{apd:proof3}

In the BPA algorithm, inequalities \eqref{eqn:NiT}\eqref{equ:1}\eqref{equ:2}\eqref{equ:3} still hold for any time slot $t=1,2,..,T$, with $U_i(t)=\frac{\sigma_0}{\sqrt{1+N_i(t)}}\Phi^{-1}(1-1/\sqrt{2\pi e}t^2)$. The estimated mean $\hat{\mu}_i(t)$ is a Gaussian random variable with mean $\frac{\mu_i^0+N_i(t)\mu_i}{1+N_i(t)}$ and variance $\frac{N_i(t)\sigma_0^2}{(1+N_i(t))^2}$. The proof then follows the similar steps as Appendix \ref{apd:proof1}, with inequality \eqref{eqn:P} written as
\begin{equation*}
\begin{aligned}
&\mathbb{P}\{\hat{\mu}_{i^*}(t)\leqslant \mu_{i^*}-U_{i^*}(t)\}\leqslant \\
&\mathbb{P}\left\{z\geqslant \sqrt{\frac{N_{i^*}+1}{N_{i^*}}}\Phi^{-1}\left(1-\frac{1}{\sqrt{2\pi e}\tau_{fk}^2}\right)-\frac{\Delta m_{i^*}}{\sigma_0\sqrt{N_{i^*}(t)}}\right\}.
\end{aligned}
\end{equation*}
Thus, inequalities \eqref{equ:1pro} and \eqref{equ:2pro} become
\begin{eqnarray}
\mathbb{P}\{\eqref{equ:1} \text{ holds}, t>\tau_1\}\geqslant \frac{1}{2}e^{\frac{3\Delta m_{i^*}^2}{2\sigma_0^2}}t^{-\frac{9}{8}},\quad\tau_1=e^{\frac{\Delta m_{i^*}^2}{3\sigma_0^2}} \nonumber
\\ \mathbb{P}\{\eqref{equ:2} \text{ holds}, t>\tau_2\}\geqslant \frac{1}{2}e^{\frac{3\Delta m_{i}^2}{2\sigma_0^2}}t^{-\frac{9}{8}}, \quad\tau_2=e^{\frac{\Delta m_{i}^2}{3\sigma_0^2}}. \nonumber
\end{eqnarray}
This leads to
\begin{equation*}
\begin{aligned}
&\mathbb{E}[N_i(T)] \leqslant \eta_i+\tau_1+\tau_2+\frac{1}{2}\left(e^{\frac{3\Delta m_{i^*}^2}{2\sigma_0^2}}+e^{\frac{3 \Delta m_{i}^2}{2\sigma_0^2}}\right)\sum\limits_{t=1}\limits^Tt^{-\frac{9}{8}} \\
 &\leqslant \lceil \frac{4\sigma_{0}^2}{\Delta_{i}^2}(\log{2\pi e}+4\log{T})-1 \rceil+e^{\frac{\Delta m_{i^*}^2}{3\sigma_0^2}}+e^{\frac{\Delta m_{i}^2}{3\sigma_0^2}}+ \\
& \hspace{10em} \frac{9}{2}e^{\frac{3\Delta m_{i^*}^2}{2\sigma_0^2}}+\frac{9}{2}e^{\frac{3\Delta m_{i}^2}{2\sigma_0^2}},
\end{aligned}
\end{equation*}
which completes the proof.

\section{Proof of Theorem \ref{thm:uipareg}}
\label{apd:proof4}

According to the Lemma 1 in \cite{Srivastava:15}, the utility function $Q_i^{UiPA}$ can be written as
\begin{equation*}
Q^{UiPA}_{i}(t)\leqslant  \bar{r}_{i}(t)+U_i(t),
\end{equation*}
with 
\begin{equation*}
U_i(t) \doteq \sqrt{\frac{\sum\limits_{\tau=1}\limits^{t}{r_{i}^2(\tau)}-\bar{r}_i^2(t)N_{i}(t)}{(N_{i}(t)-1)N_{i}(t)} (\log{2\pi e}+4\log{t})}.
\end{equation*}
Then, we can use \cite[Theorem 4]{Auer:02} to bound the expected loss. We have \eqref{eqn:uipaP1}, shown at the top of the next page, 
\begin{figure*}[!t]
\normalsize
\begin{equation}
\label{eqn:uipaP1}
\mathbb{P} \{\hat{\mu}_i(t)\geqslant\mu_i+U_i(t)\}= \mathbb{P} \left\{\frac{\hat{\mu}_i(t)-\mu_i}{\sqrt{(\sum_{\tau=1}^{t}{r_{i}^2(\tau)}-\bar{r}_i^2(t)N_{i}(t))/(N_i(t)(N_i(t)-1))}} \geqslant \sqrt{\log{2\pi e+4\log t}}\right\}
\leqslant 1/\sqrt{2\pi e}t^{-2}
\end{equation}
\hrulefill
\vspace*{4pt}
\end{figure*}
for all $N_i(t)\geqslant\log{2\pi e}/2+2\log t$. Furthermore, $\mathbb{P}\{\hat{\mu}_{i^*}(t)\geqslant\mu_{i^*}+U_{i^*}(t)\}$ can be similarly bounded. Lastly, using the Chi-squared distribution, we have
\begin{equation}
\label{eqn:uipaP2}
\begin{aligned}
& \mathbb{P}\{\mu_{i^*}<\mu_i+2U_i(t)\} = \\
& \mathbb{P}\left\{\frac{\sum_{\tau=1}^{t}{r_{i}^2(\tau)}-\bar{r}_i^2(t)N_{i}(t)}{\sigma_0^2}>\frac{ (N_i(t)-1) \Delta_i^2N_i(t)}{4\sigma_0^2(\log{2\pi e}+4\log{t})}\right\} \\
&\leqslant \mathbb{P}\left\{\frac{\sum_{\tau=1}^{t}{r_{i}^2(\tau)}-\bar{r}_i^2(t)N_{i}(t)}{\sigma_0^2}>4(N_i(t)-1)\right\} \\
&\leqslant e^{-N_i(t)/2} \\
& \leqslant(2\pi e)^{-1/4}t^{-1},
\end{aligned}
\end{equation}
and
\begin{equation}
\label{eqn:uipaN}
N_i(t)\geqslant \max\left\{\frac{16\sigma_0^2}{\Delta_i^2},\frac{1}{2}\right\}(\log{2\pi e}+4\log T).
\end{equation}
Combining \eqref{eqn:uipaP1}\eqref{eqn:uipaP2}\eqref{eqn:uipaN}, $N_i(t)$ can be bounded as
\begin{equation*}
\begin{aligned}
&  N_i(T)  \leqslant \max\left\{\frac{16\sigma_0^2}{\Delta_i^2},\frac{1}{2}\right\}(\log{2\pi e}+4\log T)+ \\
 & \sum\limits_{t=1}\limits^T\left(2/\sqrt{2\pi e}t^{-2}+(2\pi e)^{-1/4}t^{-1}\right) \\ 
 & \leqslant \frac{16\sigma_0^2}{\Delta_i^2}(\log{2\pi e}+4\log T)+ \\
 & ((2\pi e)^{-1/4}+2)\log T+\frac{\log{2\pi e}}{2}+\frac{2}{\sqrt{2\pi e}}.
\end{aligned}
\end{equation*}
This completes the proof.

\bibliographystyle{IEEEtran}
\bibliography{scMAB3}

\end{document}